\documentclass{lmcs} %%% last changed 2014-08-20
\pdfoutput=1
\usepackage[utf8]{inputenc}

\usepackage{lastpage}
\lmcsdoi{20}{1}{7}
\lmcsheading{}{\pageref{LastPage}}{}{}%
{Jan.~05,~2023}{Jan.~26,~2024}{}

\keywords{Program extraction, modified realizability, imperative programs, state monad}

\usepackage{hyperref}
\usepackage{virginialake}
\theoremstyle{plain} %\crefname{satz}{Satz}{S\"atze}
\newcommand{\pl}{\mathrm{PL}}
\newcommand{\ha}{\mathrm{HA}}
\newcommand{\plpr}{\vdash_I}
\renewcommand{\implies}{\Rightarrow}
\newcommand{\ml}{\mathrm{SL}}
\newcommand{\ma}{\mathrm{SA}}
\newcommand{\mlpr}{\vdash_S}
\newcommand{\hlpr}{\vdash_H}
\newcommand\hr[3]{\{{#1} \cdot  {#2}  \cdot  {#3}\}}
\newcommand{\axh}{\Delta_H}
\newcommand{\axm}{\Delta_S}
\newcommand{\mt}{\mathrm{ST}}
\newcommand{\mta}{\mathrm{ST}_{N}}
\newcommand{\axmt}{\Lambda_S}
\newcommand\meta{S^\omega}
\newcommand\axmeta{\Lambda_{S^\omega}}

\newcommand\stored{\mathsf{query}}
\newcommand\solved{\mathsf{return}_P}
\newcommand\sorted{\mathsf{sorted}}
\newcommand\swrite{\mathsf{write}}
\newcommand\scalc{\mathsf{calc}}
\newcommand\sread{\mathsf{read}}
\newcommand{\swap}{\mathsf{swap}}
\newcommand{\flip}{\mathsf{flip}}
\newcommand\scount{\mathsf{count}}

\newcommand{\sort}{\mathsf{sort}}
\newcommand{\psort}{\mathsf{psort}}
\newcommand{\compare}{\mathsf{comp}}

\newcommand\nat{\mathrm{Nat}}
\newcommand\bool{\mathrm{Bool}}
\newcommand\true{\mathsf{t}}
\newcommand\false{\mathsf{f}}
\newcommand\state{S}
\newcommand{\prl}{p_0}
\newcommand{\inl}{\iota_0}
\newcommand{\prr}{p_1}
\newcommand{\inr}{\iota_1}
\newcommand{\comp}{\, \circ \, }
\newcommand{\seq}{\, \ast \,}
\newcommand{\elim}[3]{\mathsf{elim} \, {#1} \, {#2} \, {#3}}
\newcommand{\ite}[3]{\mathsf{if}\, {#1} \, \mathsf{then} \, {#2} \, \mathsf{else} \, {#3}}
\newcommand{\unit}{\mathsf{skip}}
\newcommand{\zero}{\mathsf{default}}
\newcommand{\comm}{C}
\newcommand{\suc}{\mathrm{succ}}
\newcommand{\rec}[2]{\mathsf{rec}\, {#1}\, {#2}}
\newcommand{\while}[5]{\mathsf{while}_{#5}\, {#1}\, {#2}\, {#3}\, {#4}}

\newcommand{\proj}{\mathsf{proj}}
\newcommand{\pair}[1]{\langle {#1}\rangle}
\newcommand{\case}[3]{\mathsf{case}\, {#1}\, {#2}\, {#3}}
\newcommand{\statecase}[5]{\chi_{#1}\, {#2}\, {#3}\, {#4}\, {#5}}

\newcommand{\sr}[2]{{#1}\; \mathrm{sr}\; {#2}}

\newcommand{\srtype}[1]{\tau_S({#1})}

\begin{document}

\title[Proofs as stateful programs]{Proofs as stateful programs:\\ A first-order logic with abstract Hoare triples, and an interpretation into an imperative language}
%\titlecomment{{\lsuper*}OPTIONAL comment concerning the title, \eg,
%  if a variant or an extended abstract of the paper has appeared elsewhere.}
\thanks{}	%optional

% affiliations are numbered automatically with a, b, c (see below)
% use the optional argument to indicate the affiliation(s) of each author
% omit the argument if there is only one author, or only one affiliation
\author[T.~Powell]{Thomas Powell\lmcsorcid{0000-0002-2541-4678}}
%\author[B.~Name2]{Bob Name2}[a,b]
%\author[J.~Name3]{Josiah S.~Carberry\lmcsorcid{0000-0002-1825-0097}}[a]

% affiliation 1 (automatically numbered a)
\address{Department of Computer Science, University of Bath}	%optional
% write emails for all authors having that affiliation
\email{trjp20@bath.ac.uk}  %optional

% affiliation 2 (automatically numbered b)
%\address{University 2, address2}	%optional
%\email{name2@email2}  %optional

%% etc.

%% required for running head on odd and even pages, use suitable
%% abbreviations in case of long titles and many authors:

%%%%%%%%%%%%%%%%%%%%%%%%%%%%%%%%%%%%%%%%%%%%%%%%%%%%%%%%%%%%%%%%%%%%%%%%%%%

%% the abstract has to PRECEDE the command \maketitle:
%% be sure not to issue the \maketitle command twice!

\begin{abstract}
  \noindent We introduce an extension of first-order logic that comes equipped with additional predicates for reasoning about an abstract state. Sequents in the logic comprise a main formula together with pre- and postconditions in the style of Hoare logic, and the axioms and rules of the logic ensure that the assertions about the state compose in the correct way. The main result of the paper is a realizability interpretation of our logic that extracts programs into a mixed functional/imperative language. All programs expressible in this language act on the state in a sequential manner, and we make this intuition precise by interpreting them in a semantic metatheory using the state monad. Our basic framework is very general, and our intention is that it can be instantiated and extended in a variety of different ways. We outline in detail one such extension: A monadic version of Heyting arithmetic with a wellfounded while rule, and conclude by outlining several other directions for future work.
\end{abstract}

\maketitle

%% start the paper here:

%%%%%%%%%%%%%%%%%%%%%%%%%%%%%%%%%%%%%%%%%%%%%%%%%%%%%%%%%%%%%%%%%%%%%%%%%%%%%%%%%%%%%%%%%%
%%%%%%%%%%%%%%%%%%%%%%%%%%%%%%%%%%%%%%%%%%%%%%%%%%%%%%%%%%%%%%%%%%%%%%%%%%%%%%%%%%%%%%%%%%
\section{Introduction}
\label{sec:intro}
%%%%%%%%%%%%%%%%%%%%%%%%%%%%%%%%%%%%%%%%%%%%%%%%%%%%%%%%%%%%%%%%%%%%%%%%%%%%%%%%%%%%%%%%%%
%%%%%%%%%%%%%%%%%%%%%%%%%%%%%%%%%%%%%%%%%%%%%%%%%%%%%%%%%%%%%%%%%%%%%%%%%%%%%%%%%%%%%%%%%%

The Curry-Howard correspondence lies at the heart of theoretical computer science. Over the years, a multitude of different techniques for extracting programs from proofs have been developed, the majority of which translate formal proof systems into lambda calculi. As such, programs extracted from proofs are typically conceived as pure functional programs. 

Everyday programmers, on the other hand, often think and write in an imperative paradigm, in terms of instructions that change some underlying global state. This is reinforced by the fact that many of the most popular programming languages, including C and Python, lean towards this style. Imperative programs are nevertheless highly complex from a mathematical perspective, and while systems such as Hoare logic \cite{hoare:69:logic} or separation logic \cite{reynolds:02:separation} have been designed to reason about them, the formal extraction of imperative programs from proofs has received comparatively little attention.

In this paper, we propose a new idea in this direction, developing a formal system $\ml$ that enriches ordinary first-order logic with Hoare triples for reasoning about an abstract global state. Sequents will have the form $\Gamma\vdash \hr{\alpha}{A}{\beta}$, where $A$ is a formula and $\alpha,\beta$ assertions about the state, and proofs in the logic will include both ordinary introduction and elimination rules for predicate logic, together with special rules for reasoning about the state. We then construct a stateful realizability interpretation (based on Kreisel's modified realizability \cite{kreisel:59:analysis:finitetype}) that relates formulas in $\ml$ to terms in a mixed functional/imperative language $\mt$. Our main result is a soundness theorem, which confirms that whenever a formula is provable in $\ml$, we can extract a corresponding stateful realizing term in $\mt$. While our initial soundness theorem focuses on pure predicate logic, we subsequently show that it can be extended to arithmetic, where in particular we are then able to extract programs that contain both recursion and controlled while loops.

We are not the first to adapt traditional methods to extract imperative programs: A major achievement in this direction, for example, is the monograph \cite{poernomo-etal:05:book}, which sets up a variant of intuitionistic Hoare logic alongside a realizability translation into a standard imperative language. Other relevant examples include \cite{atkey:parametrised:09,chlipala:etal:09:imperative,danvy-filinski:89:functional,filliatre:03:nonfunctional,nanevski:etal:07:hoare,yoshida:etal:07:state}. However, these and almost all other prior work in this direction tend to focus on formal verification, with an eye towards using proof interpretations as a method for the synthesis of correct-by-construction software. In concrete terms, this means that the formal systems tend to be quite detailed and oriented towards program analysis, while the starting point is typically a program for which we want to construct a verification proof, rather than a proof from which we hope to extract a potentially unfamiliar program.

Our approach, on the other hand, is much more abstract, with an emphasis on potential applications in logic and proof theory. Our basic system $\ml$ makes almost no assumptions about the structure of the state and what we are allowed to do with it. Rather, we focus on producing a general framework for reasoning about `stateful formulas', which can then be instantiated with additional axioms to model concrete scenarios. The simplicity and generality of our framework is its most important feature, and we consider this work to be a first step towards a number of potentially interesting applications. For this reason, we include not only an extension of our system to a monadic theory of arithmetic, but conclude by sketching out some additional ways in which we conjecture that our logic and interpretation could be used and expanded, including the computational semantics of proofs and probabilistic logic.

We take ideas from three main sources. The first is a case study of Berger et al. \cite{berger-etal:14:imperative}, in which a realizability interpretation is used to extract a version of in-place quicksort, and where the imperative nature of the extracted program is presented in a semantic way using the state monad. While their program behaves imperatively ``by-chance'', terms extracted from our logic are forced to be imperative, and thus our framework offers one potential solution to their open problem of designing a proof calculus which only yields imperative programs. Indeed, an implementation of the insert sort algorithm is formally extracted in Section \ref{sec:arithmetic} below. Our second source of inspiration is the thesis of Birolo \cite{birolo:12:phd}, where a general monadic realizability interpretation is defined and then used to give an alternative, semantic presentation of learning-based interactive realizability \cite{aschieri:11:phd,berardi-liguoro:09:interactive}. However, our work goes beyond this in that it also involves a monadic extension of the target logic, whereas Birolo's applies to standard first-order logic. Finally, a number of ideas are taken from the author's previous work \cite{powell:18:state} on extracting stateful programs using the Dialectica interpretation. While there the state is used in a very specific and restricted way, here we use an analogous call-by-value monadic translation on terms.

It is important to stress that we do not claim that our work represents an optimal or complete method for extracting imperative programs from proofs, nor do we claim that it is superior to alternative methods, including the aforementioned works in the direction of verification, or, for instance, techniques based on Krivine's classical realizability \cite{krivine:09:realizability}, which could be viewed as imperative in nature. We simply offer what we consider to be a new and interesting perspective that emphasises abstraction and simplicity, and propose that our framework could prove valuable in a number of different contexts.

%%%%%%%%%%%%%%%%%%%%%%%%%%%%%%%%%%%%%%%%%%%%%%%%%%%%%%%%%%%%%%%%%%%%%%%%%%%%%%%%%%%%%%%%%%
\subsubsection*{Overview of the paper}
\label{sec:intro:overview}
%%%%%%%%%%%%%%%%%%%%%%%%%%%%%%%%%%%%%%%%%%%%%%%%%%%%%%%%%%%%%%%%%%%%%%%%%%%%%%%%%%%%%%%%%%

The main technical work that follows involves the design of three different systems, a realizability interpretation that connects them, and an instantiation of this framework in the setting of first-order arithmetic, namely:\medskip
\begin{itemize}

	\item A novel extension $\ml$ of predicate logic with abstract Hoare triples, which can be extended with additional axioms for characterising the state (Section \ref{sec:formal}).\smallskip
	
	\item A standard calculus $\mt$ for lambda terms with imperative commands, which can again be extended with additional constants for interacting with the state (Section \ref{sec:mixed}).\smallskip
	
	\item A metalanguage $\meta$ into which both $\ml$ and $\mt$ can be embedded (Section \ref{sec:meta}), which is used to formulate the realizability relation and prove its soundness (Section \ref{sec:realizability}).\smallskip
	
	\item An instantiation of $\ml$ as a theory of arithmetic, with programs extracted into an extension of $\mt$ with recursion and while loops (Section \ref{sec:arithmetic}).\medskip

\end{itemize}
Concrete examples are given, and potential applications surveyed in Section \ref{sec:future}.

%%%%%%%%%%%%%%%%%%%%%%%%%%%%%%%%%%%%%%%%%%%%%%%%%%%%%%%%%%%%%%%%%%%%%%%%%%%%%%%%%%%%%%%%%%
%%%%%%%%%%%%%%%%%%%%%%%%%%%%%%%%%%%%%%%%%%%%%%%%%%%%%%%%%%%%%%%%%%%%%%%%%%%%%%%%%%%%%%%%%%
\section{The system $\ml$: First-order logic with state}
\label{sec:formal}
%%%%%%%%%%%%%%%%%%%%%%%%%%%%%%%%%%%%%%%%%%%%%%%%%%%%%%%%%%%%%%%%%%%%%%%%%%%%%%%%%%%%%%%%%%
%%%%%%%%%%%%%%%%%%%%%%%%%%%%%%%%%%%%%%%%%%%%%%%%%%%%%%%%%%%%%%%%%%%%%%%%%%%%%%%%%%%%%%%%%%

We begin by introducing our target theory $\ml$ from which stateful programs will be extracted. This is an extension of ordinary first-order logic in the sense that the latter can always be embedded into $\ml$ (we will make this precise in Proposition \ref{stateful:extension} below). Ultimately, we are interested not so much in $\ml$ on its own, but in theories of the form $\ml+\axh+\axm$, where $\axh$ and $\axm$ are collections of (respectively non-computational and computational) axioms that together characterise the state. Several concrete examples will be given to illustrate this, and in Section \ref{sec:arithmetic} we present a variant of $\ml$ that represents a theory of first-order arithmetic with state.

Before defining $\ml$, we give a standard presentation of first-order intuitionistic predicate logic $\pl$, which serves as an opportunity to fix our basic style of formal reasoning. The language of $\pl$ consists of the logical constants $\wedge,\vee,\implies,\forall,\exists,\top,\bot$, variables $x,y,z,\ldots$, along with function symbols $f,g,h,\ldots$ and predicate symbols $P,Q,R,\ldots$, each with a fixed arity. We assume the existence of at least one constant $c$. Terms are built from variables and function symbols as usual, and formulas are built from prime formulas $P(t_1,\ldots,t_n)$, $\top$ and $\bot$ using the logical constants. We use the usual abbreviation $\neg A:\equiv A\implies \bot$. We work in a sequent style natural deduction calculus, where sequents have the form $\Gamma\plpr A$ for some context $\Gamma$ and formula $A$, and a context is a set of labelled assumptions of the form $A_1^{u_1},\ldots,A_n^{u_n}$ for pairwise distinct labels $u_i$. The axioms and rules of $\pl$ are as in Figure 1.

\begin{figure}[htp]
\label{fig:pl}
\caption{Axioms and rules of $\pl$}
\begin{equation*}
\begin{gathered}
\textbf{Propositional logic}
\\[2mm]
\Gamma\plpr A \ \ \ \mbox{if $A^u\in\Gamma$ for some $u$} 
\qquad
\Gamma\plpr \top
\\[2mm]
\vlinf{}{\wedge I}{\Gamma\plpr A\wedge B}{\Gamma\plpr A \ \ \ \Gamma\plpr B}
\qquad
\vlinf{}{\wedge E_L}{\Gamma\plpr A}{\Gamma\plpr A\wedge B}
\qquad
\vlinf{}{\wedge E_R}{\Gamma\plpr B} {\Gamma\plpr A\wedge B}
\\[2mm]
\vlinf{}{\vee I_L}{\Gamma\plpr A\vee B}{\Gamma\plpr A}
\qquad
\vlinf{}{\vee I_R}{\Gamma\plpr A\vee B}{\Gamma\plpr B}
\qquad
\vlinf{}{\vee E}{\Gamma\plpr C}{\Gamma\plpr A\vee B \ \ \ \Gamma,A^u\plpr C \ \ \ \Gamma, B^v\plpr C}
\\[2mm]
\vlinf{}{\implies I}{\Gamma \plpr A\implies B}{\Gamma,A^u\plpr B}
\qquad
\vlinf{}{\implies E}{\Gamma\plpr B}{\Gamma\plpr A\implies B \ \ \ \Gamma\plpr A}
\qquad
\vlinf{}{\bot E}{\Gamma\plpr A}{\Gamma\plpr \bot}
\\[2mm]
\textbf{Quantifier rules}
\\[2mm]
\vlinf{}{\forall I}{\Gamma\plpr \forall x A}{\Gamma\plpr A[y/x]} 
\qquad
\vlinf{}{\forall E}{\Gamma\plpr A[t/x]}{\Gamma\plpr \forall xA}
\\[2mm]
\vlinf{}{\exists I}{\Gamma\plpr \exists x A}{\Gamma\plpr A[t/x]}
\qquad
\vlinf{}{\exists E}{\Gamma\plpr C}{\Gamma\plpr \exists x A \ \ \ \Gamma,A[y/x]^u\plpr C} 
\\[2mm]
\mbox{for $\forall I$, $y\equiv x$ or $y$ not free in $A$, and $y$ not free in $\Gamma$}
\\
\mbox{for $\exists E$, $y\equiv x$ or $y$ not free in $A$, and $y$ not free in $C$ or $\Gamma$.}
\end{gathered}
\end{equation*}
\end{figure}

%%%%%%%%%%%%%%%%%%%%%%%%%%%%%%%%%%%%%%%%%%%%%%%%%%%%%%%%%%%%%%%%%%%%%%%%%%%%%%%%%%%%%%%%%%
\subsection{Stateful first-order logic}
\label{sec:formal:monadic}
%%%%%%%%%%%%%%%%%%%%%%%%%%%%%%%%%%%%%%%%%%%%%%%%%%%%%%%%%%%%%%%%%%%%%%%%%%%%%%%%%%%%%%%%%%

We now define our new logical system $\ml$, which is an extension of ordinary first-order logic with new state propositions. To be more precise, we extend the language of $\pl$ with a ternary operation $\hr{-}{-}{-}$, together with special state predicate symbols $p,q,r,\ldots$, which also have a fixed arity. Terms of $\ml$ are the same as those of $\pl$. On the other hand, there are two kinds of formulas in $\ml$: state formulas and main formulas. A \emph{state formula} is defined using state predicate symbols and propositional connectives as follows:\medskip
\begin{itemize}

	\item $\top$ and $\bot$ are state formulas,\smallskip

	\item if $p$ a state predicate symbol of arity $n$ and $t_1,\ldots,t_n$ are terms, then $p(t_1,\ldots,t_n)$ is a state formula,\smallskip
	
	\item if $\alpha,\beta$ are state formulas, so are $\alpha\wedge\beta$, $\alpha\vee\beta$, $\alpha\implies\beta$.\medskip

\end{itemize}
A \emph{main formula} (or just \emph{formula}) of $\ml$ is now defined as:\medskip
\begin{itemize}

	\item $\top$ and $\bot$ are formulas,\smallskip
	
	\item if $P$ is an ordinary predicate symbol of arity $n$ and $t_1,\ldots,t_n$ are terms, then $P(t_1,\ldots,t_n)$ is a formula,\smallskip
	
	\item if $A,B$ are formulas, so are $A\wedge B$, $A\vee B$ and $\exists x A$,\smallskip
	
	\item if $A,B$ are formulas and $\alpha,\beta$ \emph{state} formulas, then $A\implies \hr{\alpha}{B}{\beta}$ and $\forall x\, \hr{\alpha}{A}{\beta}$ are formulas.\medskip

\end{itemize}  
The notions of free and bound variables, along with substitution $\alpha[t/x]$ and $A[t/x]$ can be easily defined for both state and main formulas. 

Analogous to the construction of formulas, our basic proof system uses the auxiliary notion of a \emph{state proof} in order to define a main proof. A \emph{state sequent} has the form $\Gamma\hlpr \alpha$ where $\alpha$ is a state formula and $\Gamma$ a set of labelled state formulas. A proof of $\Gamma\hlpr \alpha$ in $\ml$ is built from the axioms and rules of \emph{classical propositional logic} i.e. the propositional axioms and rules as set out in Figure 1 plus the law of excluded middle $\Gamma\hlpr \alpha\vee \neg \alpha$, together with a set $\axh$ of as yet unspecified state axioms of the form $\Gamma\hlpr\alpha$.

A main sequent of $\ml$ has the form $\Gamma\mlpr \hr{\alpha}{A}{\beta}$, where $A$ is a formula and $\alpha,\beta$ state formulas, and $\Gamma$ is a set of labelled main formulas. A proof of $\Gamma\mlpr \hr{\alpha}{A}{\beta}$ in $\ml$ uses the axioms and rules given in Figure 2, together with a set $\axm$ of additional axioms.

\begin{figure}[htp]
\label{fig:ml}
\caption{Axioms and rules of $\ml$}
\begin{equation*}
\begin{gathered}
\textbf{Propositional axioms and rules}
\\[2mm]
\Gamma\mlpr \hr{\alpha}{A}{\alpha} \ \ \ \mbox{if $A^u\in\Gamma$ for some $u$} 
\qquad
\Gamma\mlpr \hr{\alpha}{\top}{\alpha}
\\[2mm]
\vlinf{}{\wedge_S I}{\Gamma\mlpr \hr{\alpha}{A\wedge B}{\gamma}}{\Gamma\mlpr \hr{\alpha}{A}{\beta} \ \ \ \Gamma\mlpr \hr{\beta}{B}{\gamma}}
\\[2mm]
\vlinf{}{\wedge_S E_L}{\Gamma\mlpr \hr{\alpha}{A}{\beta}}{\Gamma\mlpr \hr{\alpha}{A\wedge B}{\beta}}
\qquad
\vlinf{}{\wedge_S E_R}{\Gamma\mlpr \hr{\alpha}{B}{\beta}} {\Gamma\mlpr \hr{\alpha}{A\wedge B}{\beta}}
\\[2mm]
\vlinf{}{\vee_S I_L}{\Gamma\mlpr \hr{\alpha}{A\vee B}{\beta}}{\Gamma\mlpr \hr{\alpha}{A}{\beta}}
\qquad
\vlinf{}{\vee_S I_R}{\Gamma\mlpr \hr{\alpha}{A\vee B}{\beta}}{\Gamma\mlpr \hr{\alpha}{B}{\beta}}
\\[2mm]
\vlinf{}{\vee_S E}{\Gamma\mlpr \hr{\alpha}{C}{\gamma}}{\Gamma\mlpr \hr{\alpha}{A\vee B}{\beta} \ \ \ \Gamma,A^u\mlpr \hr{\beta}{C}{\gamma} \ \ \ \Gamma, B^v\mlpr \hr{\beta}{C}{\gamma}}
\\[2mm]
\vlinf{}{\implies_S I}{\Gamma \mlpr \hr{\gamma}{A\implies \hr{\alpha}{B}{\beta}}{\gamma}}{\Gamma,A^u\mlpr \hr{\alpha}{B}{\beta}}
\qquad
\vlinf{}{\implies_S E}{\Gamma\mlpr \hr{\alpha}{B}{\delta}}{\Gamma\mlpr \hr{\alpha}{A\implies \hr{\gamma}{B}{\delta}}{\beta} \ \ \ \Gamma\mlpr \hr{\beta}{A}{\gamma}}
\\[2mm]
\vlinf{}{\bot_S E}{\Gamma\mlpr\hr{\alpha}{A}{\gamma}}{\Gamma\mlpr \hr{\alpha}{\bot}{\beta}} \ \ \ \mbox{for $\gamma$ arbitrary}
\\[2mm]
\textbf{Quantifier rules}
\\[2mm]
\vlinf{}{\forall_S I}{\Gamma\mlpr \hr{\gamma}{\forall x\, \hr{\alpha}{A}{\beta}}{\gamma}}{\Gamma\mlpr \hr{\alpha[y/x]}{A[y/x]}{\beta[y/x]}} 
\qquad
\vlinf{}{\forall_S E}{\Gamma\mlpr \hr{\alpha}{A[t/x]}{\gamma[t/x]}}{\Gamma\mlpr \hr{\alpha}{\forall x\, \hr{\beta}{A}{\gamma}}{\beta[t/x]}}
\\[2mm]
\vlinf{}{\exists_S I}{\Gamma\mlpr \hr{\alpha}{\exists x A}{\beta}}{\Gamma\mlpr \hr{\alpha}{A[t/x]}{\beta}}
\qquad
\vlinf{}{\exists_S E}{\Gamma\mlpr \hr{\alpha}{C}{\gamma}}{\Gamma\mlpr \hr{\alpha}{\exists x A}{\beta} \ \ \ \Gamma,A[y/x]^u\mlpr \hr{\beta}{C}{\gamma}} 
\\[2mm]
\mbox{for $\forall_S I$, $y\equiv x$ or $y$ not free in $A$, $\alpha$, $\beta$, and $y$ not free in $\Gamma$}
\\
\mbox{for $\exists_SE$, $y\equiv x$ or $y$ not free in $A$, and $y$ not free in $C$, $\alpha$, $\beta$, $\gamma$ or $\Gamma$.}
\\[2mm]
\textbf{Basic Hoare rules}
\\[2mm]
\vlinf{}{cons}{\Gamma\mlpr\hr{\alpha}{A}{\delta}}{\alpha\hlpr\beta \ \ \ \Gamma\mlpr\hr{\beta}{A}{\gamma} \ \ \ \gamma\hlpr\delta}
\\[2mm]
\vlinf{}{cond}{\Gamma\mlpr\hr{\gamma}{A}{\delta}}{\hlpr \alpha\vee \beta \ \ \ \Gamma\mlpr\hr{\alpha\wedge\gamma}{A}{\delta} \ \ \ \Gamma\mlpr\hr{\beta\wedge\gamma}{A}{\delta}}
\\[2mm]
\textbf{Additional axioms}
\\[2mm]
\mbox{state axioms $\axh$ of the form $\Gamma\hlpr\alpha$}
\\
\mbox{main axioms $\axm$ of the form $\Gamma\mlpr\hr{\alpha}{A}{\beta}$}
\end{gathered}
\end{equation*}
\end{figure}
We now make precise what we mean when we characterise $\ml$ as an extension of standard first-order logic. The following is provable with an easy induction over derivations in $\pl$:

\begin{prop}
\label{stateful:extension}
For any formula $A$ of $\pl$ and state formula $\alpha$, define the main formula $A_\alpha$ of $\ml$ by
\begin{itemize}

	\item $Q_\alpha:=Q$ for $Q$ atomic,
	
	\item $(A\wedge B)_\alpha:=A_\alpha\wedge B_\alpha$, $(A\vee B)_\alpha:=A_\alpha\vee B_\alpha$ and $(\exists x\, A)_\alpha:=\exists x\, A_\alpha$,
	
	\item $(A\implies B)_\alpha:=A_\alpha\implies\hr{\alpha}{B_\alpha}{\alpha}$ and $(\forall x\, A)_\alpha:=\forall x\, \hr{\alpha}{A_\alpha}{\alpha}$.

\end{itemize}
Then whenever $\Gamma\plpr A$ is provable in $\pl$, we have that $\Gamma_\alpha,\Delta\mlpr\hr{\alpha}{A_\alpha}{\alpha}$ is provable in $\ml$, where $\Delta$ is arbitrary and $\Gamma_\alpha:=(A_1)^{u_1}_\alpha,\ldots,(A_n)^{u_n}_\alpha$ for $\Gamma:=A_1^{u_1},\ldots,A_n^{u_n}$.
\end{prop}

%%%%%%%%%%%%%%%%%%%%%%%%%%%%%%%%%%%%%%%%%%%%%%%%%%%%%%%%%%%%%%%%%%%%%%%%%%%%%%%%%%%%%%%%%%
\subsection{The intuition behind $\ml$}
\label{sec:formal:intuition}
%%%%%%%%%%%%%%%%%%%%%%%%%%%%%%%%%%%%%%%%%%%%%%%%%%%%%%%%%%%%%%%%%%%%%%%%%%%%%%%%%%%%%%%%%%

The intended semantic meaning of $\Gamma\hlpr\alpha$ is that $\alpha$ can be inferred from the assumptions $\Gamma$ for any fixed state. More specifically, if we imagine a semantic variant $[\alpha](\pi)$ of each state formula where now the dependency on an underlying state $\pi$ is made explicit, the semantics of $\Gamma\hlpr\alpha$ is just
\begin{equation*}
[\Gamma](\pi)\implies [\alpha](\pi)
\end{equation*}
On the other hand, the intended meaning of $\Gamma\mlpr\hr{\alpha}{A}{\beta}$ is that from assumptions $\Gamma$, if $\alpha$ holds with respect to some initial state, then we can infer that $A$ is true and $\beta$ holds with respect to some modified state, or more precisely:
\begin{equation}
\label{eqn:implication:semantics}
[\Gamma]\implies (\exists \pi\, [\alpha](\pi)\implies ([A]\wedge \exists \pi'\, [\beta](\pi')))
\end{equation}
In particular, the computational interpretation of (\ref{eqn:implication:semantics}) above will be a program that takes some input state $\pi$ satisfying $[\alpha](\pi)$ and returns a realizer-state pair $\pair{x,\pi'}$ such that $x$ realizes $A$ and $[\beta](\pi')$ holds.

Our semantic interpretation $[\cdot]$ will be properly defined in Section \ref{sec:meta}. Crucially, in $\ml$ the state is \emph{implicit}, and so there are no variables or terms of state type. The state will rather be made explicit in our metatheory $\meta$. The main axioms and rules of $\ml$ simply describe how this semantic interpretation propagates in a call-by-value manner through the usual axioms and rules of first-order logic. The state itself is brought into play through the Hoare rules along with the additional axioms $\axh$ and $\axm$. 

To give a more detailed explanation, consider the introduction rule $\wedge_S I$. Semantically, the idea is that if 
\begin{equation*}
\exists \pi\, [\alpha](\pi)\implies ([A]\wedge \exists \pi_1\, [\beta](\pi_1)) \ \ \ \mbox{and} \ \ \ \exists \pi_1\, [\beta](\pi_1)\implies ([B]\wedge \exists \pi_2\, [\gamma](\pi_2))
\end{equation*}
both hold, then we can infer
\begin{equation*}
\exists \pi\, [\alpha](\pi)\implies ([A]\wedge [B]\wedge \exists \pi_2\, [\gamma](\pi_2))
\end{equation*}
and this can be regarded as a `stateful' version of the usual conjunction introduction rule. In particular, we note that this is no longer symmetric: Informally speaking, $[A]$ is `proven first', followed by $[B]$. Under our realizability interpretation, this rule will correspond to realizer for $\hr{\alpha}{A}{\beta}$ being sequentially composed with a realizer for $\hr{\beta}{B}{\gamma}$. 

Similarly, for the stateful quantifier rule $\forall_SI$, the premise $\hr{\alpha[y/x]}{A[y/x]}{\beta[y/x]}$ corresponds semantically to
\begin{equation*}
\exists \pi\, [\alpha[y/x]](\pi)\implies [A[y/x]]\wedge \exists \pi'\, [\beta[y/x]](\pi')
\end{equation*}
and so by the standard quantifier rule of ordinary predicate logic we would have
\begin{equation*}
\forall x(\exists \pi\, [\alpha](\pi)\implies [A]\wedge \exists \pi'\, [\beta](\pi'))
\end{equation*}
From this we can directly infer
\begin{equation*}
\exists \pi\, [\gamma](\pi)\implies \forall x(\exists \pi\, [\alpha](\pi)\implies [A]\wedge \exists \pi'\, [\beta](\pi'))\wedge \exists \pi\, [\gamma](\pi)
\end{equation*}
for arbitrary $\gamma$, and the above now this has the right form, namely the semantic interpretation of $\hr{\gamma}{\forall x\, \hr{\alpha}{A}{\beta}}{\gamma}$.

We note that the traditional rules of Hoare logic correspond to certain simple cases of our rules. For example, the composition rule could be viewed as special case of $\wedge_SI$ for $A=B=\top$, more specifically the derivation:
\begin{equation*}
\vlderivation{
							\vlin{}{\wedge_SE_L}
								{\Gamma\mlpr\hr{\alpha}{\top}{\gamma}}
								{
\vliin{}{\wedge_SI}
										{\Gamma\mlpr\hr{\alpha}{\top\wedge\top}{\gamma}}
										{\vlhy{\Gamma\mlpr\hr{\alpha}{\top}{\beta}}}
										{\vlhy{\Gamma\mlpr\hr{\beta}{\top}{\gamma}}}
}
}
\end{equation*}

In a similar spirit, the two Hoare rules of $\ml$ correspond to the \emph{consequence} and \emph{conditional} rules of ordinary Hoare logic, with the traditional conditional rule falling out as a special case of ours since we assume $\Gamma\hlpr \alpha\vee \neg \alpha$. In Section \ref{sec:arithmetic} we extend this further by adding a controlled while loop to our logic. But for now, we illustrate our logic with some very straightforward scenarios.

\begin{exa}[Simple query and return]
\label{ex:readwrite:logic}
Consider a very simple state, on which we can perform the following actions:\medskip
\begin{enumerate}

	\item Store any value from our domain of discourse in some query location.\smallskip
	
	\item For the current value $x$ in the query location, return a suitable answer $y$ such that $P(x,y)$ holds for some fixed binary predicate system of the logic, and store this value. \smallskip
	
	\item Retrieve the computed value $y$ from the state.\medskip

\end{enumerate}
We formalise those three actions by including two unary state predicates $\stored$ and $\solved$, where $\stored(x)$ denotes that $x$ is currently stored in the query location, and $\solved(x)$ denotes that some $y$ satisfying $P(x,y)$ has been returned. We would then add the following axioms to $\axm$, which intuitively represent each of the above actions:\medskip
\begin{enumerate}

	\item $\Gamma\mlpr\hr{\alpha}{\top}{\stored(x)}$ where $\alpha$ ranges over all state formulas,\smallskip
	
	\item $\Gamma\mlpr\hr{\stored(x)}{\top}{\solved(x)}$\smallskip
	
	\item $\Gamma\mlpr\hr{\solved(x)}{\exists y\, P(x,y)}{\top}$\medskip

\end{enumerate}

The second rule should be regarded as representing an abstract updating of the state where the return value is stored somewhere, and the third the act of retrieval, where following the update we can now formally prove $\exists y\, P(x,y)$. Note that the computational aspects of this interpretation will only become apparent under the realizability interpretation, where a realizer for $\hr{\solved(x)}{\exists y\, P(x,y)}{\top}$ will be a function that accesses any state satisfying $\solved(x)$ and produces as output the $y$ satisfying $P(x,y)$ (See Example \ref{ex:readwrite:program}).

We can now, for example, derive the following theorem in $\ml+\axh+\axm$ (for $\axh=\emptyset$), where $\alpha,\beta$ are any state formulas:
\begin{equation*}
\label{ex:readwrite:theorem}
\mlpr\hr{\beta}{\forall x\, \hr{\alpha}{\exists y\, P(x,y)}{\top}}{\beta}
\end{equation*}
The derivation below is essentially a proof of $\forall x\exists y\, P(x,y)$ that, necessarily, utilising properties of the state:
\begin{equation*}
\vlderivation{
	\vlin{}{\forall_SI}
		{\mlpr\hr{\beta}{\forall x\, \hr{\alpha}{\exists y\, P(x,y)}{\top}}{\beta}}
		{
			\vlin{}{\wedge_SE_L}
				{\mlpr\hr{\alpha}{\exists y\, P(x,y)}{\top}}
				{
					\vliin{}{\wedge_SI}
						{\mlpr\hr{\alpha}{\top\wedge\exists y\, P(x,y)}{\top}}
						{
							\vlin{}{\wedge_SE_L}
								{\mlpr\hr{\alpha}{\top}{\solved(x)}}
								{
									\vliin{}{\wedge_SI}
										{\mlpr\hr{\alpha}{\top\wedge\top}{\solved(x)}}
										{\vlhy{\mlpr\hr{\alpha}{\top}{\stored(x)}}}
										{\vlhy{\mlpr\hr{\stored(x)}{\top}{\solved(x)}}}
								}
						}
						{
							\vlhy{\!\!\!\!\!\!\!\!\!\!\!\mlpr\hr{\solved(x)}{\exists y\, P(x,y)}{\top}}    % Keeping formula out of margin
						}
				}
		}
}
\end{equation*}
We note that while state formulas and actions are used in the \emph{proof}, if we set $\alpha=\beta=\top$ then the components of the \emph{theorem} itself are just formulas in ordinary first-order logic. A program corresponding to this derivation will be formally extracted in Example \ref{ex:readwrite:program}, but referring to the semantic explanations above, we can view each instance of $\wedge_SI$ in the proof as a sequential composition of state actions (with $\wedge_SE$ just cleaning up the central logical formula), and finally $\forall_S I$ is just the stateful version of the usual $\forall$-introduction rule.
\end{exa}

\begin{exa}[Fixed-length array sorting]
\label{ex:sort:logic}
Let us now consider our state as an array of length three, and elements in that array as having some order structure. We formalise this in $\ml$ by introducing $1,2,3$ as constants representing our three locations, along with two state predicates: a binary predicate $\leq$ for comparing elements at locations $l$ and $l'$, and a nullary predicate $\sorted$ that declares that the state is sorted. These can be characterised by adding the following axiom schemes, but to $\axh$ rather than $\axm$ as they do not represent state \emph{actions}:
\begin{equation*}
\begin{aligned}
	&\Gamma\hlpr1\leq 2\wedge 2\leq 3\implies\sorted\\
	&\Gamma\hlpr l\leq l'\vee l'\leq l \ \ \ \mbox{where $l,l'$ range over $\{1,2,3\}$}
\end{aligned}
\end{equation*}
We then allow a single action on our array, namely the swapping of a pair of elements in the list. Suppose that $\alpha$ is a state formula of the form
\begin{equation}
\label{eqn:conjform}
\alpha:\equiv l_1\leq l'_1\wedge\ldots \wedge l_n\leq l'_n
\end{equation}
where $l_i,l_i$ range over locations $\{1,2,3\}$. Now for $l,l'\in \{1,2,3\}$ let $\alpha[l\leftrightarrow l']$ denote $\alpha$ where all instances of $l$ and $l'$ are swapped, so that if e.g. $\alpha=3\leq 2\wedge 1\leq 2\wedge 1\leq 3$ then
\begin{equation*}
\alpha[2\leftrightarrow 3]=2\leq 3\wedge 1\leq 3\wedge 1\leq 2
\end{equation*}
We axiomatise the swapping of the values in locations of some arbitrary pair $l,l'\in \{1,2,3\}$ by adding to $\axm$ all instances of
\begin{equation*}
\Gamma\mlpr\hr{\alpha}{\top}{\alpha[l\leftrightarrow l']}
\end{equation*} 
where $\alpha$ ranges over state formulas of the form (\ref{eqn:conjform}). The statement that all arrays of length three can be sorted is then formulated as
\begin{equation*}
\label{ex:sort:theorem}
\mlpr\hr{\top}{\top}{\sorted}
\end{equation*}
Let us now give a proof of this statement in $\ml+\axh+\axm$. As with the previous example, we emphasise that the underlying computation represented by this proof will only become visible through the realizability interpretation. 

First, let $\alpha:=1\leq 2\wedge 1\leq 3$, and define $\mathcal{D}_1$ as
\begin{equation*}
\vlderivation{
\vliin{}{cond[2\leq 3\vee 3\leq 2]}
{\mlpr\hr{\alpha}{\top}{\sorted}}
{
\vlin{}{cons}{\mlpr\hr{2\leq 3\wedge \alpha}{\top}{\sorted}}
	{\vlhy{\mlpr\hr{2\leq 3\wedge\alpha}{\top}{2\leq 3\wedge\alpha}}}
}
{
\vlin{}{cons}{\mlpr\hr{3\leq 2\wedge \alpha}{\top}{\sorted}}{\vlin{}{2\leftrightarrow 3}{\mlpr\hr{3\leq 2\wedge \alpha}{\top}{2\leq 3\wedge 1\leq 3\wedge 1\leq 2}}{}}
}
}
\end{equation*}
where for the left instance of $cons$ we use $2\leq 3\wedge \alpha\hlpr\sorted$, in the right that $2\leq 3\wedge 1\leq 3\wedge 1\leq 2\hlpr\sorted$, and for the final instance of $cond$ we use $\hlpr 2\leq 3\vee 3\leq 2$. Now let $\mathcal{D}_2$ be defined by
\begin{equation*}
\vlderivation{
\vlin{}{\wedge_S E_L}
	{\mlpr\hr{2\leq 1\wedge 2\leq 3}{\top}{\sorted}}
	{\vliin{}{\wedge_S I}{\mlpr\hr{2\leq 1\wedge 2\leq 3}{\top\wedge \top}{\sorted}}
		{\vlin{}{1\leftrightarrow 2}{\mlpr\hr{2\leq 1\wedge 2\leq 3}{\top}{1\leq 2\wedge 1\leq 3}}{\vlhy{}}}
		{\vlin{}{}{\mlpr\hr{1\leq 2\wedge 1\leq 3}{\top}{\sorted}}{\vlhy{\mathcal{D}_1}}}
	}	
}
\end{equation*}
Then we have $\mathcal{D}_3$:
\begin{equation*}
\vlderivation{
	\vliin{}{cond[2\leq 1\vee 1\leq 2]}{\mlpr\hr{2\leq 3}{\top}{\sorted}}
		{\vlin{}{}{\mlpr\hr{2\leq 1\wedge 2\leq 3}{\top}{\sorted}}{\vlhy{\mathcal{D}_2}}}
		{\vlin{}{cons}{\mlpr\hr{1\leq 2\wedge 2\leq 3}{\top}{\sorted}}{\vlhy{\hr{1\leq 2\wedge 2\leq 3}{\top}{1\leq 2\wedge 2\leq 3}}}}
}
\end{equation*}
where here $cond$ uses $\hlpr 2\leq 1\vee 1\leq 2$, and finally
\begin{equation*}
\vlderivation{
	\vlin{}{\wedge_SE_L}
		{\mlpr\hr{\top}{\top}{\sorted}}
		{
			\vliin{}{\wedge_SI}{\mlpr\hr{\top}{\top\wedge\top}{\sorted}}
				{
					\vliin{}{cond[2\leq 3\vee 3\leq 2]}{\mlpr\hr{\top}{\top}{2\leq 3}}
			{\vlhy{\mlpr\hr{2\leq 3}{\top}{2\leq 3}}}
			{\vlin{}{2\leftrightarrow 3}{\mlpr\hr{3\leq 2}{\top}{2\leq 3}}{\vlhy{}}}	
				}
				{
				\vlin{}{}{\mlpr\hr{2\leq 3}{\top}{\sorted}}{\vlhy{\mathcal{D}_3}}
				}
		}
}
\end{equation*}
In contrast to Example \ref{ex:readwrite:logic} above, this is an example of a purely imperative proof that involves no propositional formulas other than $\top$. As we will see in Example \ref{ex:readwrite:program}, the proof corresponds to a purely imperative program. 
\end{exa}

%%%%%%%%%%%%%%%%%%%%%%%%%%%%%%%%%%%%%%%%%%%%%%%%%%%%%%%%%%%%%%%%%%%%%%%%%%%%%%%%%%%%%%%%%%
%%%%%%%%%%%%%%%%%%%%%%%%%%%%%%%%%%%%%%%%%%%%%%%%%%%%%%%%%%%%%%%%%%%%%%%%%%%%%%%%%%%%%%%%%%
\section{The system $\mt$: A simple functional/imperative term calculus}
\label{sec:mixed}
%%%%%%%%%%%%%%%%%%%%%%%%%%%%%%%%%%%%%%%%%%%%%%%%%%%%%%%%%%%%%%%%%%%%%%%%%%%%%%%%%%%%%%%%%%
%%%%%%%%%%%%%%%%%%%%%%%%%%%%%%%%%%%%%%%%%%%%%%%%%%%%%%%%%%%%%%%%%%%%%%%%%%%%%%%%%%%%%%%%%%

We now define our calculus $\mt+\axmt$ whose terms will represent realizers for proofs in $\ml+\axh+\axm$. This is a standard typed lambda calculus for mixed functional and imperative programs, and is defined to include basic terms together with additional constants in some set $\axmt$, where the latter are intuitively there to realize the axioms in $\axm$. Semantics for the terms will be given via a monadic translation into the metalanguage defined in the next section. Types are defined by the grammar 
\begin{equation*}
X ::= D \, | \, \comm \, | \, X\times X\, | \, X+X \, | \, X\to X
\end{equation*}
while basic terms are defined as
\begin{equation*}
\begin{aligned}
e ::= &\unit\, | \, \zero_X \, | \, c \, | \, f \, | \, x \, | \, \prl(e) \, | \, \prr(e) \, | \, e\comp e\, | \, \inl(e) \, |\, \inr(e)\, |\\
& \elim{e}{e}{e} \, | \, \lambda x.e \, | \, e\, e \, | \, \ite{\alpha}{e}{e}   
\end{aligned}
\end{equation*}
where $f$ ranges over all function symbols of $\ml$, $c$ are constants in $\axmt$, and $\alpha$ ranges over state formulas of $\ml$. Typing derivations of the form $\Gamma\vdash t : X$ are given in Figure 3, where $\Gamma$ is a set of typed variables. Note that the types of constants $c\in \axmt$ are also left unspecified. The type $\comm$ should be interpreted as a type of commands that act on the state but don't return any values.  \smallskip
\begin{figure}[tp]
\label{fig:mt}
\caption{Typing derivations for $\mt+\axmt$}
\begin{equation*}
\begin{gathered}
\Gamma\vdash f:D^n\to D \ \ \ \mbox{where $f$ has arity $n$}
\qquad
\Gamma\vdash c:X\\[2mm]
\qquad
\Gamma\vdash x:X \ \ \ \mbox{if $x:X$ in $\Gamma$}
\qquad
\Gamma\vdash \unit:\comm
\\[2mm]
\frac{\Gamma\vdash s:X \ \ \ \Gamma\vdash t:Y}{\Gamma\vdash s\comp t:X\times Y}
\qquad
\frac{\Gamma\vdash t:X\times Y}{\Gamma\vdash \prl(t):X}
\qquad
\frac{\Gamma\vdash t:X\times Y}{\Gamma\vdash \prr(t):Y}
\\[2mm]
\frac{\Gamma\vdash t:X}{\Gamma\vdash \inl(t):X+Y}
\qquad
\frac{\Gamma\vdash t:Y}{\Gamma\vdash \inr(t):X+Y}
\qquad
\\[2mm]
\frac{\Gamma\vdash r:X+Y \ \ \ \Gamma\vdash s:X\to Z  \Gamma\vdash t:Y\to Z}{\Gamma\vdash \elim{r}{s}{t}\vdash Z}
\\[2mm]
\frac{\Gamma,x:X\vdash t:Y}{\Gamma\vdash \lambda x.t:X\to Y}
\qquad
\frac{\Gamma\vdash t:X\to Y \ \ \ \Gamma\vdash s:X}{\Gamma\vdash ts:Y}
\qquad
\Gamma\vdash \zero_X:X
\\[2mm]
\frac{\Gamma\vdash s:X \ \ \ \Gamma\vdash t:X \ \ \ \mbox{$x:D\in\Gamma$ for all free variables of $\alpha$}}{\Gamma\vdash \ite{\alpha}{s}{t}:X}
\end{gathered}
\end{equation*}
\end{figure}

A \emph{denotational} semantics of $\mt+\axmt$, which is what we require for our realizability interpretation, will be specified in detail in Section \ref{sec:meta} below. Operationally, the idea is that terms $t:X$ of $\mt+\axmt$ take some input state $\pi$ and evaluates to some value $v$ and final state $\pi_1$ in a call-by-value way. We choose our notation to reflect the underlying stateful computations: For example, $s\circ t$ is used instead of what would normally be a pairing operation, because this plays the role of simulating composition in the stateful setting. Indeed, it will be helpful to consider a derived operator for sequential composition that also incorporates the `cleanup' seen in the examples above, and corresponds to an instance of $\wedge_SI$ followed by an instance of $\wedge_SE_L$: 
\begin{defi}
\label{def:star}
If $\Gamma\vdash s:\comm$ and $\Gamma\vdash t:X$ then $\Gamma\vdash s\seq t:=\prr(s\comp t):X$. In particular, if $\Gamma\vdash t:\comm$ then $\Gamma\vdash s\seq t:\comm$.
\end{defi} 
A full exploration of the operational semantics of $\mt+\axmt$ along with correctness with respect to the denotational semantics defined in Section \ref{sec:meta:terms} is left to future work, but we provide some further insight in Remark \ref{rem:operational} below, including a more detailed discussion of the relationship between $\circ$ and pairing.

%%%%%%%%%%%%%%%%%%%%%%%%%%%%%%%%%%%%%%%%%%%%%%%%%%%%%%%%%%%%%%%%%%%%%%%%%%%%%%%%%%%%%%%%%%
%%%%%%%%%%%%%%%%%%%%%%%%%%%%%%%%%%%%%%%%%%%%%%%%%%%%%%%%%%%%%%%%%%%%%%%%%%%%%%%%%%%%%%%%%%
\section{A monadic embedding of $\ml$ and $\mt$ into a metatheory \texorpdfstring{$\meta$}{Sω}}
\label{sec:meta}
%%%%%%%%%%%%%%%%%%%%%%%%%%%%%%%%%%%%%%%%%%%%%%%%%%%%%%%%%%%%%%%%%%%%%%%%%%%%%%%%%%%%%%%%%%
%%%%%%%%%%%%%%%%%%%%%%%%%%%%%%%%%%%%%%%%%%%%%%%%%%%%%%%%%%%%%%%%%%%%%%%%%%%%%%%%%%%%%%%%%%

We now give a semantic interpretation of both state formulas of $\ml+\axh+\axm$ and terms in $\mt+\axmt$ into a standard higher-order, many sorted logic $\meta+\axmeta$. 

%%%%%%%%%%%%%%%%%%%%%%%%%%%%%%%%%%%%%%%%%%%%%%%%%%%%%%%%%%%%%%%%%%%%%%%%%%%%%%%%%%%%%%%%%%
\subsection{The system \texorpdfstring{$\meta$}{Sω}}
\label{sec:meta:theory}
%%%%%%%%%%%%%%%%%%%%%%%%%%%%%%%%%%%%%%%%%%%%%%%%%%%%%%%%%%%%%%%%%%%%%%%%%%%%%%%%%%%%%%%%%%

This logic contains typed lambda terms along with equational axioms for reasoning about them, together with the usual axioms and rules of many-sorted predicate logic. Because most aspects of the logic are completely standard, and in any case it is purely a verifying system, we are less detailed in specifying it. Types are defined as follows:
\begin{equation*}
X::=D\, | \, 1 \, | \, \bool \, | \, \state \, | \, X\times X\, | \, X\to X
\end{equation*}
where $D$ represents objects in the domain of $\ml$ (just as in $\mt$), $\bool$ a type of booleans, and states are now explicitly represented as objects of type $\state$. Our metatheory is an equational calculus, with an equality symbol $=_X$ for all types. Typed terms include:\medskip
\begin{itemize}

	\item variables $x,y,z,\ldots$ for each type, where we denote state variables by $\pi,\pi_1,\pi_2,\ldots$\smallskip
	
	\item a constant $f:D^n\to D$ for each $n$-ary function symbol of $\ml$,\smallskip
	
	\item additional, as yet unspecified constant symbols $c:X$ for interpreting objects in $\axmt$, along with axioms that characterise them,\smallskip
	
	\item a unit element $():1$ along with the axiom $x=()$,\smallskip
	
	\item boolean constants $\true$ and $\false$, with the axiom $x=_\bool \true\vee x=_\bool\false$,\smallskip
	
	\item pairing $\pair{s,t}$ and projection $\proj_0 (t)$, $\proj_1 (t)$ operators, with the usual axioms,\smallskip
	
	\item terms formed by lambda abstraction and application, with the rule $(\lambda x.t)s=t[s/x]$,\smallskip
	
	\item for each type $X$ a case operator $\case{(b)}{(s)}{(t)}$ for $b:\bool$ and $s,t:X$, with axioms $\case{\false}{x}{y}=x$ and $\case{\true}{x}{y}=y$.\medskip

\end{itemize}
We sometimes write $x^X$ instead of $x:X$, and we use abbreviations such as $\pair{x,y,z}$ for $\pair{x,\pair{y,z}}$. Atomic formulas of $\meta$ include all ordinary predicate symbols $P,Q,R,\ldots$ of $\ml$ as atomic formulas, where an $n$-ary predicate $P$ in $\ml$ takes arguments of type $D^n$ in $\meta$, along with predicates $p,q,r,\ldots$ for each state predicate symbol of $\ml$, but now, if $p$ is an $n$-ary state predicate in $\ml$, $p$ takes arguments of type $D^n\times\state$ in $\meta$. General formulas are built using the usual logical connectives, including quantifiers for all types. The axioms and rules of $\meta$ include the axioms of rules of predicate logic (now in all finite types), axioms for the terms, along with the usual equality axioms (including full extensionality). Because $\meta$ acts as a verifying theory, we freely use strong axioms (such as extensionality), without concerning ourselves with the minimal such system that works. 

%%%%%%%%%%%%%%%%%%%%%%%%%%%%%%%%%%%%%%%%%%%%%%%%%%%%%%%%%%%%%%%%%%%%%%%%%%%%%%%%%%%%%%%%%%
\subsection{The embedding \texorpdfstring{$[\cdot]$}{[⋅]} on state formulas of $\ml$}
\label{sec:meta:formulas}
%%%%%%%%%%%%%%%%%%%%%%%%%%%%%%%%%%%%%%%%%%%%%%%%%%%%%%%%%%%%%%%%%%%%%%%%%%%%%%%%%%%%%%%%%%

The main purpose of our metalanguage is to allow us to reason semantically about $\ml$ and $\mt$. To do this, we introduce an embedding of state formulas of $\ml$ and terms of $\mt$ into $\meta$. We use the same notation $[\cdot]$ for both, as there is no danger of ambiguity. An informal explanation of the meaning of $[\cdot]$ on formulas is given in Section \ref{sec:formal:intuition}. 

An important point to highlight here is that under the semantics, the arity of state formulas change: For any state formula $\alpha$, the interpreted formula $[\alpha]$ now contains a single additional free variable $\pi:S$ representing the underlying state. As mentioned earlier, state is implicit in our logic and term languages, but needs to be made explicit under the semantics. 
\begin{defi}
\label{def:interpret:terms}
For each term $t$ of $\ml$, there is a natural interpretation of $t$ as a term of type $D$ in $\mt$, namely $x\mapsto x:D$ and $f(t_1,\ldots,t_n)\mapsto f(t_1\circ\cdots\circ t_n):D$. Similarly, there is a natural interpretation of $t$ into $\meta$, this time with $f(t_1,\ldots,t_n)\mapsto f(\pair{t_1,\ldots,t_n})$. We use the same notation for $t$ in each of the three systems, as there is no risk of ambiguity.
\end{defi}
\begin{defi}
\label{def:emb:formulas}
For each \emph{state} formula $\alpha$ of $\ml$, we define a formula $[\alpha](\pi)$ of $\meta$, whose free variables are the same as those of $\alpha$ (but now typed with type $D$) with the potential addition of a single state variable $\pi$, as follows:\medskip
\begin{itemize}

	\item $[\top](\pi):=\top$ and $[\bot](\pi):=\bot$,\smallskip
	
	\item $[p(t_1,\ldots,t_n)](\pi):=p(t_1,\ldots,t_n,\pi)$,\smallskip
	
	\item $[\alpha\wedge\beta](\pi):=[\alpha](\pi)\wedge [\beta](\pi)$, and similarly for $\alpha\vee\beta$ and $\alpha\implies \beta$.\medskip

\end{itemize}
\end{defi}
The following Lemma is easily proven using induction over propositional derivations. 
\begin{lem}
\label{lem:state:emb:sound}
If $\Gamma\hlpr\alpha$ in $\ml$ then $[\alpha](\pi)$ is provable in $\meta$ from the assumptions $[\Gamma](\pi)$, where $[\Gamma](\pi):=[\alpha_1](\pi),\ldots,[\alpha_n](\pi)$ for $\Gamma:=\alpha_1,\ldots,\alpha_n$. This extends to proofs in $\ml+\axh$ provided that the embedding of any axiom in $\axh$ is provable in $\meta+\axmeta$.
\end{lem}
We are now in a position to make the semantic meaning of main formulas of $\ml$ precise. Note that, technically speaking, this is not necessary in what follows, neither to formulate our realizability interpretation nor to prove our soundness theorem. This is because our main realizability relation (Definition \ref{def:realizability:relation}) is of the form $(\sr{x}{A})$ for main formulas $A$ of $\ml$: This relation is basically equivalent to $(x\, \mathrm{mr}\, [A])$ for $[A]$ as defined below and a standard modified realizability relation mr, but our realizability interpretation essentially acts as a simultaneous inductive definition of both standard realizability and the embedding $[\cdot]$, and so neither of the latter need to be separately defined.

\begin{defi}
\label{asd}
\emph{main} formula $A$ of $\sl$, we define a formula $[A]$ of $\meta$, whose free variables are the same as those of $A$ (but now typed with type $D$), as follows:\medskip
\begin{itemize}

	\item $[\top]:=\top$ and $[\bot]:=\bot$,\smallskip
	
	\item $[P(t_1,\ldots,t_n)]:=P(t_1,\ldots,t_n)$,\smallskip
	
	\item $[A\wedge B]:=[A]\wedge [B]$, $[A\vee B]:=[A]\vee [B]$ and $[\exists x\, A]:=\exists x^D\, [A]$,\smallskip
	
	\item $[A\implies\hr{\alpha}{B}{\beta}]:=[A]\implies [\hr{\alpha}{B}{\beta}]$ and $[\forall x\, \hr{\alpha}{A}{\beta}]:=\forall x^D\, [\hr{\alpha}{A}{\beta}]$\medskip

\end{itemize}
where $[\hr{\alpha}{A}{\beta}]:=\exists \pi^S\, [\alpha](\pi)\implies [A]\wedge \exists \pi'\, [\beta](\pi')$. 
\end{defi}

Similarly to Lemma \ref{lem:state:emb:sound}, we can now prove the following by induction over derivations in $\ml$. We omit the proof, because it is straightforward and in any case not necessary in what follows.
\begin{prop}
\label{prop:state:emb:sound}
If $\Gamma\mlpr\hr{\alpha}{A}{\beta}$ in $\ml$ then $[\hr{\alpha}{A}{\beta}]$ is provable in $\meta$ from the assumptions $[\Gamma]$, where $[\Gamma]:=[A_1],\ldots,[A_n]$ for $\Gamma:=A_1,\ldots,A_n$. This extends to proofs in $\ml+\axh+\axm$ provided that the embedding of any axiom in $\axh$ and $\axm$ is provable in $\meta+\axmeta$.
\end{prop}

%%%%%%%%%%%%%%%%%%%%%%%%%%%%%%%%%%%%%%%%%%%%%%%%%%%%%%%%%%%%%%%%%%%%%%%%%%%%%%%%%%%%%%%%%%
\subsection{The embedding \texorpdfstring{$[\cdot]$}{[⋅]} on terms of $\mt$}
\label{sec:meta:terms}
%%%%%%%%%%%%%%%%%%%%%%%%%%%%%%%%%%%%%%%%%%%%%%%%%%%%%%%%%%%%%%%%%%%%%%%%%%%%%%%%%%%%%%%%%%

Our translation on terms is a call-by-value monadic translation using the state monad $S\to X\times S$, which, intuitively speaking, gives a denotational interpretation of a standard call-by-value operational semantics of terms of $\mt$. A full treatment of the corresponding operational semantics and its adequacy with respect to the and denotational semantics will be left to future work, as we only require the latter for the realizability interpretation. However, to help motivate the definitions that follow, and also justify our claim that $[t]$ can be viewed as an imperative program, we give an informal intuition in Remark \ref{rem:operational} below. 

We first define a translation on types of $\mt$ as follows:\medskip
\begin{itemize}

	\item $[D]:=D$, $[\comm]:=1$ and $[X\times Y]:=[X]\times [Y]$,\smallskip
	
	\item $[X+Y]:=\bool\times [X]\times [Y]$\smallskip
	
	\item $[X\to Y]:=[X]\to S\to [Y]\times S$\medskip

\end{itemize}
\begin{lem}
\label{lem:zero}
For any type $X$ of $\ml$, the type $[X]$ is inhabited, in the sense that we can define a canonical closed term $0_X:[X]$.
\end{lem}

\begin{proof}
Induction on types, letting $0_D:=c$ for a constant symbol which is assumed to exist in $\ml$. The only other nonstandard case is $0_{X\to Y}$, which can be defined as $\lambda x,\pi\, . \, \pair{0_Y,\pi}$.
\end{proof}
Finally, before introducing our translation on terms, we need to add characteristic functions to $\meta$ for all state formulas (analogous to the characteristic functions for quantifier-free formulas in \cite{gerhardy-kohlenbach:05:herbrand}). For any state formula $\alpha[x_1,\ldots,x_n]$ of $\ml$, where $x_1,\ldots,x_n$ are the free variables of $\alpha$, we introduce constants $\statecase{\alpha}{}{}{}:D^n\to S\to X\to X\to X$ satisfying the axioms
\begin{equation*}
\begin{aligned}
[\alpha][x_1,\ldots,x_n](\pi)&\implies \statecase{\alpha}{\pair{x_1,\ldots,x_n}}{\pi}{y}{z}=y \\
[\neg\alpha][x_1,\ldots,x_n](\pi)&\implies \statecase{\alpha}{\pair{x_1,\ldots,x_n}}{\pi}{y}{z}=z 
\end{aligned}
\end{equation*}
\begin{defi}
\label{def:emb:terms}
For each term $\Gamma\vdash t:X$ of $\mt$ we define a term $[\Gamma]\vdash [t]:S\to [X]\times S$ of $\meta$ as follows, where $[\cdot]$ is defined on contexts as $[x_1:X_1,\ldots,x_n:X_n]:=x_1:[X_1],\ldots,x_n:[X_n]$:\medskip
\begin{itemize}

	\item $[x]\pi:=\pair{x,\pi}$,\smallskip
	
	\item $[\unit]\pi:=\pair{(),\pi}$,\smallskip
	
	\item $[f]\pi:=\pair{\lambda x^{D^n},\pi\, . \, \pair{fx,\pi},\pi}$,\smallskip
	
	\item $[c]\pi$ is appropriately defined for each additional constant in $\axmt$,\smallskip
	
	\item $[s\circ t]\pi:=\pair{a,b,\pi_2}$ where $\pair{a,\pi_1}:=[s]\pi$ and $\pair{b,\pi_2}:=[t]\pi_1$,\smallskip
	
	\item $[\prl t]\pi:=\pair{a,\pi_1}$ and $[\prr t]\pi:=\pair{b,\pi_1}$ where $\pair{a,b,\pi_1}:=[t]\pi$,\smallskip
	
	\item $[\inl t]\pi:=\pair{\false,a,0_Y,\pi_1}$ and $[\inr t]:=\pair{\true,0_X,b,\pi_1}$ for $\pair{a,\pi_1}:=[t]\pi$,\smallskip
	
	\item $[\elim{r}{s}{t}]\pi:=\case{e}{(f a \pi_2)}{(g b \pi_3)}$ for $\pair{e,a,b,\pi_1}:=[r]\pi$, $\pair{f,\pi_2}:=[s]\pi_1$, $\pair{g,\pi_3}:=[t]\pi_1$,\smallskip
	
	\item $[\lambda x.t]\pi:=\pair{\lambda x^{[X]}.[t],\pi}$,\smallskip
	
	\item $[ts]\pi:=fa\pi_2$ for $\pair{f,\pi_1}:=[t]\pi$ and $\pair{a,\pi_2}:=[s]\pi_1$,\smallskip
	
	\item $[\zero_X]\pi:=\pair{0_X,\pi}$,\smallskip
	
	\item $[\ite{\alpha[x_1,\ldots,x_n]}{s}{t}]\pi:=\statecase{\alpha}{\pair{x_1,\ldots,x_n}}{\pi}{([s]\pi)}{([t]\pi)}$ where $\{x_1,\ldots,x_n\}$ are the free variables of $\alpha$.\medskip

\end{itemize}
\end{defi}
\begin{rem}
\label{rem:operational}
The intuition behind the embedding $[t]$ is to give a denotational semantics for $t$ as a stateful program. Informally, a term $t:X$ of $\mt$ would have operational behaviour $\pair{t,\pi}\Downarrow \pair{v,\pi_1}$, where we imagine that $t$ in state $\pi$ evaluates to a value $v$ and returns a state $\pi_1$. We view $[X]$ as the set of denotations of values of type $X$, and $[t]:S\to [X]\times S$ accordingly as a function with $[t]\pi=\pair{[v],\pi_1}$. The interpretation of fully typed terms $[\Gamma\vdash t:X]$ correspond to mappings $[\Gamma]\to S\to [X]\times S$ where free variables of $t$ can be instantiated by values of the corresponding type.

With this intuition in mind, each component of Definition \ref{def:emb:terms} corresponds to a natural \emph{operational} interpretation of the corresponding term forming rule. For example, the intended operational semantics of $s\circ t$ would be
\begin{equation*}
\frac{\pair{s,\pi}\Downarrow \pair{u,\pi_1} \ \ \ \pair{t,\pi_1}\Downarrow \pair{v,\pi_2}}{\pair{s\comp t,\pi}\Downarrow \pair{\pair{u,v},\pi_2}}
\end{equation*}
and this behaviour is embodied by the denotation $[s\comp t]$, which maps $\pi$ to $\pair{a,b,\pi_2}$ where $a$ can be interpreted as the denotation of the value $u$, and $b$ as the denotation of $v$. Similarly, the call-by-value operational semantics of function application would be expressed by the rule
\begin{equation*}
\frac{\pair{t,\pi}\Downarrow \pair{u,\pi_1} \ \ \ \pair{s,\pi_1}\Downarrow \pair{v,\pi_2} \ \ \ \pair{uv,\pi_2}\Downarrow \pair{w,\pi_3}}{\pair{ts,\pi}\Downarrow \pair{w,\pi_3}}
\end{equation*}
i.e. we evaluate first the function $t$, then the argument $s$, then the function application itself. This order of evaluation along with the behaviour of the state is represented semantically by $[ts]$. Finally, we note that the main interactions with the state for terms of $\mt+\axmt$ are driven by the constants $c\in\axmt$, whose semantic interpretation as stateful programs will need to be specified in each case.

Our monadic semantics and its relationship with the intended call-by-value operational semantics is very similar in spirit to the monadic denotational semantics used in \cite{danner-licata-ramyaa:15:cost} and related papers (but with the state monad instead of the complexity monad). We leave a formal definition and detailed exploration of the operational semantics of extracted imperative programs to future work, where it is anticipated that existing work on the operational semantics of imperative languages could be useful and revealing (e.g. \cite{churchill-laird-mccusker:13:imperative,mccusker:10:imperative,reddy:96:state}), particularly in extending our realizability interpretation to incorporate richer imperative languages.
\end{rem}
The following lemmas will be useful when verifying our realizability interpretation in the next section. The first is by a simple induction on terms.
\begin{lem}
\label{lem:emb:terms}
For any term $t$ of $\ml$, we have $[t]\pi=\pair{t,\pi}$ (cf. Definitions \ref{def:interpret:terms} and \ref{def:emb:terms}).
\end{lem}

\begin{lem}[Currying in $\mt$]
\label{lem:currying}
Suppose that $\Gamma,x:X,y:Y\vdash t:Z$ is a term in $\mt$, and define $\Gamma\vdash\lambda^\ast v.t:X\times Y\to Z$ by $\lambda^\ast v.t:=\lambda v.(\lambda x,y.t)(\prl v)(\prr v)$ where $v$ is not free in $t$. Then for any $s:X\times Y$ we have
\begin{equation*}
[(\lambda^\ast v.t)s]\pi=[t][a/x,b/y]\pi_1
\end{equation*}
where $\pair{a,b,\pi_1}:=[s]\pi$.
\end{lem}

\begin{proof}
By unwinding the definition of $[\cdot]$. For any variable $v:X\times Y$ we have $[\prl v]\pi=\pair{\proj_0v,\pi}$ and $[\prr v]\pi=\pair{\proj_1v,\pi}$, and we also have $[\lambda x,y\,.\,t]\pi=\pair{\lambda x,\pi.\pair{\lambda y.[t],\pi},\pi}$. We therefore calculate
\begin{equation*}
[(\lambda x,y.t)(\prl v)]\pi=(\lambda x,\pi.\pair{\lambda y.[t],\pi})(\proj_0 v)\pi=\pair{\lambda y.[t][\proj_0 v/x],\pi}
\end{equation*}
and thus
\begin{equation*}
[(\lambda x,y.t)(\prl v)(\prr v)]\pi=(\lambda y.[t][\proj_0v/x])(\proj_1v)\pi=[t][\proj_0v/x,\proj_1v/y]\pi
\end{equation*}
Finally, we can see that if $\pair{a,b,\pi_1}:=[s]\pi$ then
\begin{equation*}
\begin{aligned}
[(\lambda^\ast v.t)(s)]\pi&=(\lambda v.[(\lambda x,y.t)(\prl v)(\prr v)])(\pair{a,b})\pi_1\\
&=(\lambda v.[t][\proj_0v/x,\proj_1v/y])(\pair{a,b})\pi_1\\
&=[t][\proj_0v/x,\proj_1v/y][\pair{a,b}/v]\pi_1\\
&=[t][a/x,b/y]\pi_1
\end{aligned}
\end{equation*}
which completes the proof.
\end{proof}

%%%%%%%%%%%%%%%%%%%%%%%%%%%%%%%%%%%%%%%%%%%%%%%%%%%%%%%%%%%%%%%%%%%%%%%%%%%%%%%%%%%%%%%%%%
%%%%%%%%%%%%%%%%%%%%%%%%%%%%%%%%%%%%%%%%%%%%%%%%%%%%%%%%%%%%%%%%%%%%%%%%%%%%%%%%%%%%%%%%%%
\section{A realizability interpretation of $\ml$ into $\mt$}
\label{sec:realizability}
%%%%%%%%%%%%%%%%%%%%%%%%%%%%%%%%%%%%%%%%%%%%%%%%%%%%%%%%%%%%%%%%%%%%%%%%%%%%%%%%%%%%%%%%%%
%%%%%%%%%%%%%%%%%%%%%%%%%%%%%%%%%%%%%%%%%%%%%%%%%%%%%%%%%%%%%%%%%%%%%%%%%%%%%%%%%%%%%%%%%%

We now come to the main contribution of the paper, which is the definition of a realizability relation between terms of $\mt$ and formulas of $\ml$, along with a soundness theorem that shows us how to extract realizers from proofs. Our metatheory $\meta$ is used to define the realizability relation and prove the soundness theorem. 
\begin{defi}[Types of realizers]
\label{def:realizability:types}
To each main formula $A$ of $\ml$ we assign a type $\srtype{A}$ of $\mt$ as follows:\medskip
\begin{itemize}

	\item $\srtype{\top}=\srtype{\bot}=\srtype{P(t_1,\ldots,t_n)}:=C$,\smallskip
	
	\item $\srtype{A\wedge B}:=\srtype{A}\times\srtype{B}$,\smallskip
	
	\item $\srtype{A\vee B}:=\srtype{A}+\srtype{B}$,\smallskip
	
	\item $\srtype{\exists x\, A}:=D\times \srtype{A}$,\smallskip
	
	\item $\srtype{A\implies \hr{\alpha}{B}{\beta}}:=\srtype{A}\to\srtype{B}$,\smallskip
	
	\item $\srtype{\forall x\, \hr{\alpha}{A}{\beta}}:=D\to \srtype{A}$.\medskip

\end{itemize}
\end{defi}

\begin{defi}[Realizability relation]
\label{def:realizability:relation}
For each main formula $A$ of $\ml$ we define a formula $\sr{x}{A}$ of $\meta$, whose free variables are contained in those of $A$ (now typed with type $D$) together with a fresh variable $x:[\srtype{A}]$, by induction on the structure of $A$ as follows:\medskip
\begin{itemize}

	\item $\sr{x}{Q}:=Q$ for $Q=\top,\bot$ or $P(t_1,\ldots,t_n)$,\smallskip
	
	\item $\sr{x}{A\wedge B}:=(\sr{\proj_0x}{A})\wedge (\sr{\proj_1x}{B})$,\smallskip
	
	\item $\sr{x}{A\vee B}:=(\proj_0x=\false\implies \sr{\proj_0(\proj_1x)}{A})\wedge (\proj_0x=\true\implies \sr{\proj_1(\proj_1x)}{B})$,\smallskip
	
	\item $\sr{x}{\exists y\, A(y)}:=(\sr{\proj_1x}{A})[\proj_0x/y]$,\smallskip
	
	\item $\sr{f}{(A\implies\hr{\alpha}{B}{\beta})}:=\forall x^{[\srtype{A}]}\, (\sr{x}{A}\implies \sr{fx}{\hr{\alpha}{B}{\beta}})$,\smallskip
	
	\item $\sr{f}{(\forall x\, \hr{\alpha(x)}{A(x)}{\beta(x)})}:=\forall x^D\, (\sr{fx}{\hr{\alpha(x)}{A(x)}{\beta(x)}})$,

\end{itemize}\medskip
where for $x:S\to [\srtype{A}]\times S$ we define
\begin{itemize}

	\item $\sr{x}{\hr{\alpha}{A}{\beta}}:=\forall \pi^S\, ([\alpha](\pi)\implies \sr{\proj_0(x\pi)}{A}\wedge [\beta](\proj_1(x\pi)))$.

\end{itemize}
\end{defi}
The following substitution lemma is easily proven by induction on formulas of $\ml$.
\begin{lem}
\label{lem:realizability:subs}
For any term $t$ of $\ml$ and $s:[\srtype{A}]$ we have $\sr{s}{A[t/x]}=(\sr{s}{A})[t/x]$, where $x$ is not free in $s$ and on the right hand side we implicitly mean the natural interpretation of $t$ in $\meta$ (cf. Definition \ref{def:interpret:terms}).
\end{lem}

\begin{thm}[Soundness]
\label{thm:main}
Suppose that 
\begin{equation*}
\Gamma:=A_1^{u_1},\ldots,A_n^{u_n}\mlpr\hr{\alpha}{A}{\beta}
\end{equation*}
is provable in $\ml$. Then we can extract from the proof a term $\Delta,\srtype{\Gamma}\vdash t:\srtype{A}$ of $\mt$, where $\Delta$ contains the free variables of $\Gamma$ and $\hr{\alpha}{A}{\beta}$ (typed with type $D$) and $\srtype{\Gamma}:=x_1:\srtype{A_1},\ldots,x_n:\srtype{A_n}$ for fresh variables $x_1,\ldots,x_n$, such that the formula
\begin{equation*}
\sr{[t]}{\hr{\alpha}{A}{\beta}}
\end{equation*}
is provable in $\meta$ from the assumptions $(\sr{x_1}{A_1})^{u_1},\ldots,(\sr{x_n}{A_n})^{u_n}$ for $x_i:[\srtype{A_i}]$. The theorem holds more generally for proofs in $\ml+\axh+\axm$, now provably in $\meta+\axmeta$, if:\medskip
\begin{itemize}

	\item for any axiom $\Gamma\hlpr\alpha$ in $\axh$, the corresponding axiom $[\Gamma](\pi)\implies[\alpha](\pi)$ is added to $\axmeta$, \smallskip

	\item for any axiom in $\axm$ there is a term $t$ of $\mt+\axmt$ such that $[t]$ realizes that axiom provably in $\meta+\axmeta$. 
	
\end{itemize}		
\end{thm}

\begin{proof}
Induction on the structure of derivations in $\ml$. In all cases, we assume as global assumptions $(\sr{x_1}{A_1})^{u_1},\ldots,(\sr{x_n}{A_n})^{u_n}$, and our aim is then to produce a term $t$ such that if $[\alpha](\pi)$ holds for some state variable $\pi$, then $\sr{a}{A}$ and $[\beta](\pi_1)$ hold for $\pair{a,\pi_1}:=[t]\pi$.\medskip
\begin{itemize}

	\item For the axiom $\Gamma\mlpr\hr{\alpha}{A}{\alpha}$, if $A^u\in\Gamma$ we define $t:=x$ for the corresponding variable $x:\srtype{A}$. Then $[x]\pi:=\pair{x,\pi}$ for $\sr{x}{A}$ and $[\alpha](\pi)$.  For $\Gamma\mlpr\hr{\alpha}{\top}{\alpha}$ we define $t:=\unit$ and the verification is even simpler.
	
	\medskip
	
	\item ($\wedge_SI$) Given terms $s,t$ with $\sr{[s]}{\hr{\alpha}{A}{\beta}}$ and $\sr{[t]}{\hr{\beta}{B}{\gamma}}$, from $[\alpha](\pi)$ we can infer $\sr{a}{A}$ and $[\beta](\pi_1)$ for $\pair{a,\pi_1}:=[s]\pi$, and from $[\beta](\pi_1)$ it follows that $\sr{b}{B}$ and $[\gamma](\pi_2)$ for $\pair{b,\pi_2}:=[t]\pi_1$, therefore we have shown that $\sr{[s\circ t]}{\hr{\alpha}{A\wedge B}{\gamma}}$.
	
	\medskip
	
	\item ($\wedge_SE_i$) If $\sr{[t]}{\hr{\alpha}{A\wedge B}{\beta}}$ then $\sr{\pair{a,b}}{A\wedge B}$ and $[\beta](\pi_1)$ follow from $[\alpha](\pi)$, where $\pair{a,b,\pi_1}:=[t]\pi$. But then $\sr{[\prl t]}{\hr{\alpha}{A}{\beta}}$ and $\sr{[\prr t]}{\hr{\alpha}{B}{\beta}}$.
	
	\medskip
	
	\item ($\vee_SI_i$) If $\sr{[t]}{\hr{\alpha}{A}{\beta}}$ and $[\alpha](\pi)$ holds, then $\sr{a}{A}$ and $[\beta](\pi_1)$ for $\pair{a,\pi_1}:=[t]\pi$, and therefore 
	\begin{equation*}
	(b=\false\implies \sr{a}{A})\wedge (b=\true\implies \sr{0_{\srtype{B}}}{B})
	\end{equation*}
	for $b:=\false$. Thus $\sr{[\inl t]}{A\vee B}$. By an entirely analogous argument we can show that $\sr{[\inr t]}{A\vee B}$ whenever $\sr{[t]}{B}$.
	
	\medskip
	
	\item ($\vee_SE$) Suppose that $r$, $s(x)$ and $t(y)$ are such that $\sr{[r]}{\hr{\alpha}{A\vee B}{\beta}}$, $\sr{[s](x)}{\hr{\beta}{C}{\gamma}}$ assuming $\sr{x}{A}$, and $\sr{[t](y)}{\hr{\beta}{C}{\gamma}}$ assuming $\sr{y}{B}$. We claim that 
	\begin{equation*}
	\sr{[\elim{r}{(\lambda x.s)}{(\lambda y.t)}]}{\hr{\alpha}{C}{\gamma}}
	\end{equation*}
	To prove this, first note that if $[\alpha](\pi)$, we have $\sr{\pair{e,a,b}}{A\vee B}$ and $[\beta](\pi_1)$ for $\pair{e,a,b,\pi_1}\!:=[r]\pi$. There are now two possibilities. If $e=\false$ then
	\begin{equation*}
	\begin{aligned}
	[\elim{r}{(\lambda x.s)}{(\lambda y.t)}]\pi&=fa\pi_2 \ \ \ \mbox{for $\pair{f,\pi_2}:=[\lambda x.s]\pi_1=\pair{\lambda x.[s](x),\pi_1}$}\\
	&=(\lambda x.[s](x))a\pi_1\\
	&=[s](a)\pi_1
	\end{aligned}
	\end{equation*}
	But since $[\beta](\pi_1)$ holds and $e=\false$ also implies that $\sr{a}{A}$, we have $\sr{c}{C}$ and $[\gamma](\pi_2)$ for $\pair{c,\pi_2}:=[s](a)\pi_1$, which proves the main claim in the case $e=\false$. An analogous argument works for the case $e=\true$.
	
	\medskip
	
	\item ($\implies_SI$) If $t(x)$ is such that $\sr{[t](x)}{\hr{\alpha}{B}{\beta}}$ whenever $\sr{x}{A}$, then by definition we have $$\sr{\lambda x.[t]}{A\implies\hr{\alpha}{B}{\beta}}$$ and therefore $\sr{[\lambda x.t]}{\hr{\gamma}{A\implies\hr{\alpha}{B}{\beta}}{\gamma}}$ for any $\gamma$.
	
	\medskip
	
	\item ($\implies_SE$) Assume that $\sr{[s]}{\hr{\beta}{A}{\gamma}}$ and $\sr{[t]}{\hr{\alpha}{A\implies\hr{\gamma}{B}{\delta}}{\beta}}$. If $[\alpha](\pi)$ holds then defining $\pair{f,\pi_1}:=[t]\pi$ we have $[\beta]\pi_1$ and
	\begin{equation*}
	\sr{x}{A}\implies \sr{fx}{\hr{\gamma}{B}{\delta}}
	\end{equation*}
	Similarly, defining $\pair{a,\pi_2}:=[s]\pi_1$, it follows that $[\gamma](\pi_2)$ and $\sr{a}{A}$. Finally, setting $\pair{b,\pi_3}:=fa\pi_2$ it follows that $\sr{b}{B}$ and $[\delta](\pi_3)$, and we have therefore proven that $\sr{[ts]}{\hr{\alpha}{B}{\delta}}$.
	
	\medskip
	
	\item ($\bot_SE$) If $\sr{[t]}{\hr{\alpha}{\bot}{\beta}}$ then from $[\alpha](\pi)$ we can infer $\sr{a}{\bot}$ and $[\beta](\pi_2)$ for $\pair{a,\pi_1}:=[t]\pi$. But $\sr{a}{\bot}=\bot$, and from $\bot$ we can deduce anything, and in particular $\sr{0_{\srtype{A}}}{A}$ and $[\gamma](\pi)$, from which it follows that $\sr{[\zero_{\srtype{A}}]}{\hr{\alpha}{A}{\gamma}}$.
	
	\medskip
	
	\item ($\forall_SI$) Suppose that $t(x)$ is such that $\sr{[t](y)}{\hr{\alpha[y/x]}{A[y/x]}{\beta[y/x]}}$, where $y\equiv x$ or $y$ is not free in $\hr{\alpha}{A}{\beta}$, and $y$ is not free in $\Gamma$. Then since $y$ is not free in any of the assumptions $\sr{x_i}{A_i}$, we can deduce in $\meta$ that 
		\begin{equation*}
		\forall x^D\, \sr{[t](x)}{\hr{\alpha}{A}{\beta}}
		\end{equation*}
	and therefore $\sr{\lambda x.[t]}{\forall x\, \hr{\alpha}{A}{\beta}}$, and thus (just as for $\implies_SI$) we have 
	\begin{equation*}
	\sr{[\lambda x.t]}{\hr{\gamma}{\forall x\, \hr{\alpha}{A}{\beta}}{\gamma}}
	\end{equation*}
	for any $\gamma$.
		
	\medskip
		
	\item ($\forall_SE$) Suppose that $\sr{[s]}{\hr{\alpha}{\forall x\, \hr{\beta}{A}{\gamma}}{\beta[t/x]}}$ and that $[\alpha](\pi)$ holds. Then we have $\sr{f}{\forall x\, \hr{\beta}{A}{\gamma}}$ and $[\beta][t/x](\pi_1)$ for $\pair{f,\pi}:=[s]\pi$. Now, using Lemma \ref{lem:emb:terms} we have $[st]\pi=ft\pi_1$ for the natural interpretation of $t$ in $\meta$, since we can prove in $\meta$ that
	\begin{equation*}
	\sr{ft}{\hr{\beta[t/x]}{A[t/x]}{\gamma[t/x]}}
	\end{equation*}
	it follows that $\sr{a}{A[t/x]}$ and $[\gamma][t/x](\pi_2)$ for $\pair{a,\pi_2}:=ft\pi_1$, and therefore we have shown that $\sr{[st]}{\hr{\alpha}{A[t/x]}{\gamma[t/x]}}$.
	
	\medskip
	
	\item ($\exists_SI$) If $\sr{[s]}{\hr{\alpha}{A[t/x]}{\beta}}$ and $[\alpha](\pi)$ then $\sr{a}{A[t/x]}$ and $[\beta](\pi_1)$ for $\pair{a,\pi_1}:=[s]\pi$. By Lemma \ref{lem:realizability:subs} we therefore have $(\sr{a}{A})[t/x]$, and therefore $\sr{\pair{t,a}}{\exists x\, A}$. Observing (using Lemma \ref{lem:emb:terms}) that $[t\circ s]\pi=\pair{t,a,\pi_1}$, we have shown that $\sr{[t\circ s]}{\hr{\alpha}{\exists x\, A}{\beta}}$.
	
	\medskip
	
	\item ($\exists_SE$) Suppose that $s$ and $t(x,z)$ are such that $\sr{[s]}{\hr{\alpha}{\exists x\, A}{\beta}}$ and
	\begin{equation*}
	\sr{z}{A[y/x]}\implies \sr{[t](y,z)}{\hr{\beta}{C}{\gamma}}
	\end{equation*}
	where $y\equiv x$ or $y$ is not free in $A$, and $y$ is also not free in $C$, $\alpha$, $\beta$, $\gamma$ or $\Gamma$. By Lemma \ref{lem:realizability:subs} that $\sr{z}{A[y/x]}=(\sr{z}{A})[y/x]=\sr{\pair{y,z}}{\exists x\, A}$ we therefore have
	\begin{equation*}
	\sr{\pair{y,z}}{\exists x\, A}\implies \sr{[t](y,z)}{\hr{\beta}{C}{\gamma}}
	\end{equation*}
	Now, applying Lemma \ref{lem:currying} to $\Delta,\Gamma,y:D,z:\srtype{A}\vdash t:\srtype{C}$, we have
	\begin{equation*}
	[(\lambda^\ast v.t)s]\pi=[t](e,a)\pi_1
	\end{equation*}
	for $\pair{e,a,\pi_1}:=[s]\pi$. Now, if $[\alpha](\pi)$ holds, then we have $\sr{\pair{e,a}}{\exists x\, A}$ and $[\beta](\pi_1)$, and therefore since $\sr{[t](e,a)}{\hr{\beta}{C}{\gamma}}$, we have $\sr{c}{C}$ and $[\gamma](\pi_2)$ for $\pair{c,\pi_2}=[t](e,a)\pi_1=[(\lambda^\ast v.t)s]\pi$, and thus we have shown that $\sr{[(\lambda^\ast v.t)s]}{\hr{\alpha}{C}{\gamma}}$.
	
	\medskip
	
	\item ($cons$) If $\alpha\hlpr\beta$ and $\gamma\hlpr\delta$ then by Lemma \ref{lem:state:emb:sound} both $[\alpha](\pi)\implies[\beta](\pi)$ and $[\gamma](\pi)\implies [\delta](\pi)$ are provable in $\meta$ (respectively $\meta+\axmeta$ for the general version of the theorem) for any $\pi:S$. It is then easy to show that if $\sr{[t]}{\hr{\beta}{A}{\gamma}}$ then we also have $\sr{[t]}{\hr{\alpha}{A}{\delta}}$. 
	
	\medskip
	
		\item ($cond$) Suppose that $\sr{[s]}{\hr{\alpha\wedge\gamma}{A}{\delta}}$ and $\sr{[t]}{\hr{\beta\wedge\gamma}{A}{\delta}}$. We claim that
	\begin{equation*}
	\sr{[\ite{\alpha}{s}{t}]}{\hr{\gamma}{A}{\delta}}
	\end{equation*}
	To prove this, suppose that $[\gamma](\pi)$ holds. Since $\hlpr\alpha\vee \beta$ then $[\alpha](\pi)\vee [\beta](\pi)$ is provable in $\meta$, and so we consider two cases. Let $\{x_1,\ldots,x_n\}$ be the free variables of $\alpha$. If $[\alpha](\pi)$ holds, then
	\begin{equation*}
	[\ite{\alpha}{s}{t}]\pi=\statecase{\alpha}{\pair{x_1,\ldots,x_n}}{\pi}{([s]\pi)}{([t]\pi)}=[s]\pi
	\end{equation*}
	and since then $[\alpha](\pi)\wedge [\gamma](\pi)$ we have $\sr{a}{A}$ and $[\delta](\pi_1)$ for $\pair{a,\pi_1}:=[s]\pi$. On the other hand, if $[\beta](\pi)$ holds, then by an analogous argument we can show that $\sr{a}{A}$ and $[\delta](\pi_1)$ for $\pair{a,\pi_1}:=[t]\pi=[\ite{\alpha}{s}{t}]\pi$, and we are done.

\end{itemize}\medskip
The extension of the soundness theorem to $\ml+\axh+\axm$ is straightforward, as the soundness proof is modular and so any axioms along with their realizers can be added. The first condition is needed so that Lemma \ref{lem:state:emb:sound} (needed for the $cons$ rule) continues to apply.

For the free variable condition that the free variables of $t$ are contained in those of $\Gamma$, $\hr{\alpha}{A}{\beta}$ and  $\srtype{\Gamma}$, if this were not the case, we could simply ground those variables with a canonical constant $c:D$ and we would still have $\sr{\tilde t}{\hr{\alpha}{A}{\beta}}$ for the resulting term $\tilde t$.
\end{proof}

\begin{cor}[Program extraction]
\label{cor:extraction}
Suppose that the sentence
\begin{equation*}
\mlpr\hr{\alpha}{\forall x\, \hr{\beta}{\exists y\, P(x,y)}{\gamma(x)}}{\beta}
\end{equation*}
is provable in $\ml+\axm$. Then we can extract a closed realizing term $t:D\to D\times C$ in $\mt+\axmt$ such that defining $g:D\to S\to D\times S$ by $gx\pi:=\pair{a,\pi_2}$ for $\pair{f,\pi_1}:=[t]\pi$ and $\pair{a,(),\pi_2}:=fx\pi_1$, we have
\begin{equation*}
\forall \pi^S([\alpha](\pi)\implies \forall x^D\, (P(x,\proj_0(gx\pi))\wedge [\gamma](x)(\proj_1(gx\pi))))
\end{equation*}
provably in $\meta+\axmeta$.
\end{cor}

%%%%%%%%%%%%%%%%%%%%%%%%%%%%%%%%%%%%%%%%%%%%%%%%%%%%%%%%%%%%%%%%%%%%%%%%%%%%%%%%%%%%%%%%%%
\subsection{Simplification and removal of unit types}
%%%%%%%%%%%%%%%%%%%%%%%%%%%%%%%%%%%%%%%%%%%%%%%%%%%%%%%%%%%%%%%%%%%%%%%%%%%%%%%%%%%%%%%%%%

In presentations of modified realizability that use product types instead of type sequences, it is common to introduce the notion of a Harrop formula (a formula that does not contain disjunction or existential quantification in a positive position) and define realizability in a way that all Harrop formulas have unit realizability type, so that e.g. $\srtype{\forall x\, (P\wedge Q)}=1$ for atomic predicates $P$ and $Q$, rather than $\srtype{\forall x\, (P\wedge Q)}=D\to 1\times 1$ as for us. We have avoided this simplification earlier on, as it would have added additional cases and bureaucracy to our soundness theorem. However, we can compensate retroactively for this choice by introducing equivalences on types that eliminate unit types, namely the closure under contexts of
\begin{equation*}
1\times X\simeq X\simeq X\times 1 \ \ \ (1\to X)\simeq X \ \ \ (X\to 1) \simeq 1
\end{equation*}
along with corresponding equivalences on terms, also closed under contexts:
\begin{equation*}
t^{1\times X}\simeq \proj_1(t)^X \ \ \ t^{X\times 1}\simeq\proj_0(t)^X \ \ \ t^{1\to X}\simeq t() \ \ \ t^X\simeq \lambda x^1.t \ \ \ t^{X\to 1}\simeq()
\end{equation*}
For example, in Corollary \ref{cor:extraction} we would then have 
\begin{equation*}
[t]\pi:(D\to S\to D\times 1\times S)\times S\simeq (D\to S\to D\times S)\times S \ \ \ \mbox{and} \ \ \ gx\pi\simeq fx\pi_1.
\end{equation*}
For us, the equivalence relation $\simeq$ will not play a formal role in the paper, but will be used to provide simplified descriptions of extracted programs. In particular, we do not rely on it in any way for the main results. To incorporate the simplifications formally into our framework, the most obvious route would involve defining a more elaborate realizability relation as indicated above, with the interpretation of Harrop formulas treated separately (as in the recent and related paper \cite{berger-tsuiki:21:fixedpoint}), which would then also generate many new (but routine) cases in the soundness proof. Given that we are already introducing a new and nonstandard realizability relation, in this article we prefer to work with a simple interpretation and use the equivalences above in an informal way to more concisely describe the meaning of extracted programs.

%%%%%%%%%%%%%%%%%%%%%%%%%%%%%%%%%%%%%%%%%%%%%%%%%%%%%%%%%%%%%%%%%%%%%%%%%%%%%%%%%%%%%%%%%%
\subsection{Examples of program extraction}
%%%%%%%%%%%%%%%%%%%%%%%%%%%%%%%%%%%%%%%%%%%%%%%%%%%%%%%%%%%%%%%%%%%%%%%%%%%%%%%%%%%%%%%%%%

We now continue the short illustrative examples we outlined in Section \ref{sec:formal:intuition}.

\begin{exa}[Simple read-write]
\label{ex:readwrite:program}
In Example \ref{ex:readwrite:logic} we considered a state where three actions were possible (writing to the state, performing a calculation, and reading the output from the state). We can formalise these three actions semantically in the metatheory $\meta$ by including three constants in $\axmeta$, namely $c_1:D\to S\to S$, $c_2:S\to S$ and $c_3:S\to D$, along with the characterising axioms:\medskip
\begin{enumerate}

	\item $\stored(x,c_1x\pi)$,\smallskip
	
	\item $\stored(x,\pi)\implies \solved(x,c_2\pi)$,\smallskip
	
	\item $\solved(x,\pi)\implies P(x,c_3\pi)$.

\end{enumerate}\medskip
While we are able to use these constants to form terms in $\meta$ such as $\lambda \pi,\pi_1,x\, . \,\pair{c_1x\pi,c_2\pi_1}$, which could be viewed as non-sequential in the sense that we take two input states as arguments, we can force them to be applied in a sequential, call-by-value manner by adding three corresponding constants to our term calculus $\mt$, namely including $\swrite:D\to C$, $\scalc:C$ and $\sread:D\times C$ in $\axmt$, along with the embedding rules\medskip
\begin{itemize}

	\item $[\swrite]\pi:=\pair{\lambda x,\pi'\, . \, \pair{(),c_1x\pi'},\pi}\simeq \pair{c_1,\pi}$,\smallskip
	
	\item $[\scalc]\pi:=\pair{(),c_2\pi}$ so that $[\scalc]\simeq c_2$,\smallskip
	
	\item $[\sread]\pi:=\pair{c_3\pi,(),\pi}\simeq\pair{c_3\pi,\pi}$.

\end{itemize}\medskip
and then restricting out attention to terms of the form $[t]$ for $t\in\mt+\{\swrite,\scalc,\sread\}$. We can then prove the following in $\meta$ i.e. that all axioms in $\axm$ can be realised:\medskip
\begin{itemize}

	\item $\sr{[\swrite(x)]}{\hr{\alpha}{\top}{\stored(x)}}$,\smallskip
	
	\item $\sr{[\scalc]}{\hr{\stored(x)}{\top}{\solved(x)}}$,\smallskip
	
	\item $\sr{[\sread]}{\hr{\solved(x)}{\exists y\, P(x,y)}{\top}}$.\medskip

\end{itemize}
and thus Theorem \ref{thm:main} applies to $\ml+\axh+\axm$ for $\axh=\emptyset$. In particular, we have
\begin{equation*}
\sr{[t]}{\hr{\beta}{\forall x\, \hr{\alpha}{\exists y\, P(x,y)}{\top}}{\beta}}
\end{equation*}
for $t:=\lambda x\, . \, ((\swrite(x)\seq\scalc)\seq\sread)$ where $\seq$ is sequential composition operator from Definition \ref{def:star}. A formal derivation of this term from the corresponding proof given in Example \ref{ex:readwrite:logic} is as follows:
\begin{equation*}
\vlderivation{
	\vlin{}{\forall_SI}
		{\vdash\lambda x\,.\,((\swrite(x)\seq \scalc)\seq \sread):D\to D\times C}
		{
			\vlin{}{\wedge_SE_L}
				{x:D\vdash(\swrite(x)\seq \scalc)\seq \sread:D\times C}
				{
					\vliin{}{\wedge_SI}
						{x:D\vdash (\swrite(x)\seq \scalc)\circ \sread: C\times D\times C}
						{
							\vlin{}{\wedge_SE_L}
								{x:D\vdash\swrite(x)\seq\scalc:C}
								{
									\vliin{}{\wedge_SI}
										{x:D\vdash \swrite(x)\circ\scalc:C\times C}
										{\vlhy{x:D\vdash\swrite(x):C}}
										{\vlhy{x:D\vdash\scalc:C}}
								}
						}
						{
							\vlhy{x:D\vdash\sread:D\times C}
						}
				}
		}
}
\end{equation*}
\end{exa}

\begin{exa}[Fixed-length array sorting]
\label{ex:sort:program}
In Example \ref{ex:sort:logic} we considered a situation where we are allowed a single action on our state, namely to swap elements. Analogously to the previous example, we can formalise this in our semantic environment $\meta$ by adding to $\axmeta$ constants $c_{l,l'}:S\to S$ for each pair $l,l'\in \{1,2,3\}$ along with the axiom
\begin{equation*}
[\alpha](\pi)\implies [\alpha[l\leftrightarrow l']](c_{l,l'}\pi) 
\end{equation*}
ranging over state formulas $\alpha$ of the form (\ref{eqn:conjform}) and locations $l,l'\in \{1,2,3\}$ of $\ml$, together with axioms corresponding to those of $\axh$ i.e.
\begin{equation*}
\begin{aligned}
[1\leq 2\wedge 2\leq 3](\pi)\implies\sorted(\pi) \ \ \ \mbox{and} \ \ \ [l\leq l'\vee l'\leq l](\pi)
\end{aligned}
\end{equation*}
Similarly, for each $l,l'\in \{1,2,3\}$ we add a term $\swap_{l,l'}: C$ to $\axmt$ and define $[\swap_{l,l'}]\pi:=\pair{(),c_{l,l'}\pi}$ so that
\begin{equation*}
\sr{\swap_{l,l'}}{\hr{\alpha}{\top}{\alpha[l\leftrightarrow l']}}
\end{equation*}
A derivation of a closed term $t:C$ of $\mt+\{\swap_{l,l'}\}$ such that $\sr{[t]}{\hr{\top}{\top}{\sorted}}$ is given below. In particular, we can prove in $\meta$ that $\forall \pi^S\, \sorted(\proj_1([t]\pi))$, and so the term $\lambda \pi\, . \,\proj_1([t]\pi):S\to S$ acts as a sorting program for arrays of length three. For an extracted term $t$ corresponding to the proof given in Example \ref{ex:sort:logic}, first we interpret $\mathcal{D}_1$ as
\begin{equation*}
\vlderivation{
\vliin{}{cond[2\leq 3\vee 3\leq 2]}
{\vdash t_1:=\ite{(2\leq 3)}{(\unit)}{(\swap_{2,3})}:C}
{
\vlin{}{cons}{\vdash\unit:C}
	{\vlhy{\vdash\unit:C}}
}
{
\vlin{}{cons}{\vdash\swap_{2,3}:C}{\vlin{}{2\leftrightarrow 3}{\vdash\swap_{2,3}:C}{}}
}
}
\end{equation*}
and define $t_1:=\ite{(2\leq 3)}{(\unit)}{(\swap_{2,3})}$. Now $\mathcal{D}_2$ is interpreted as
\begin{equation*}
\vlderivation{
\vlin{}{\wedge_S E_L}
	{\vdash t_2:=\swap_{1,2}\seq t_1:C}
	{\vliin{}{\wedge_S I}{\vdash \swap_{1,2}\circ t_1:C\times C}
		{\vlin{}{1\leftrightarrow 2}{\vdash\swap_{1,2}:C}{\vlhy{}}}
		{\vlin{}{}{\vdash t_1:C}{\vlhy{\mathcal{D}_1}}}
	}	
}
\end{equation*}
where we define $t_2:=\swap_{1,2}\seq t_1:C$. Continuing, $\mathcal{D}_3$ is interpreted as:
\begin{equation*}
\vlderivation{
	\vliin{}{cond[2\leq 1\vee 1\leq 2]}{t_3:=\ite{(2\leq 1)}{t_2}{(\unit)}:C}
		{\vlin{}{}{\vdash t_2:C}{\vlhy{\mathcal{D}_2}}}
		{\vlin{}{cons}{\vdash\unit:C}{\vlhy{\vdash\unit:C}}}
}
\end{equation*}
where $t_3:=\ite{(2\leq 1)}{t_2}{(\unit)}$, and finally
\begin{equation*}
\vlderivation{
	\vlin{}{\wedge_SE_L}
		{\vdash t:=(\ite{(2\leq 3)}{(\unit)}{(\swap_{2,3})})\seq t_3:C}
		{
			\vliin{}{\wedge_SI}{\vdash (\ite{(2\leq 3)}{(\unit)}{(\swap_{2,3})})\circ t_3:C\times C}
				{
					\vliin{}{cond[2\leq 3\vee 3\leq 2]}{\vdash\ite{(2\leq 3)}{(\unit)}{(\swap_{2,3})}:C}
			{\vlhy{\vdash\unit:C}}
			{\vlin{}{2\leftrightarrow 3}{\vdash\swap_{2,3}:C}{\vlhy{}}}	
				}
				{
				\vlin{}{}{\vdash t_3:C}{\vlhy{\mathcal{D}_3}}
				}
		}
}
\end{equation*}
\end{exa}

%%%%%%%%%%%%%%%%%%%%%%%%%%%%%%%%%%%%%%%%%%%%%%%%%%%%%%%%%%%%%%%%%%%%%%%%%%%%%%%%%%%%%%%%%%
%%%%%%%%%%%%%%%%%%%%%%%%%%%%%%%%%%%%%%%%%%%%%%%%%%%%%%%%%%%%%%%%%%%%%%%%%%%%%%%%%%%%%%%%%%
\section{An extension to arithmetic}
\label{sec:arithmetic}
%%%%%%%%%%%%%%%%%%%%%%%%%%%%%%%%%%%%%%%%%%%%%%%%%%%%%%%%%%%%%%%%%%%%%%%%%%%%%%%%%%%%%%%%%%
%%%%%%%%%%%%%%%%%%%%%%%%%%%%%%%%%%%%%%%%%%%%%%%%%%%%%%%%%%%%%%%%%%%%%%%%%%%%%%%%%%%%%%%%%%

We now present an extension of our framework to a stateful version of first-order intuitionistic arithmetic. On the logic side, we will add not only a stateful induction rule, but also a Hoare-style while rule for iteration over the natural numbers. On the computational side, these will be interpreted by stateful recursion in all finite types, along with a controlled while loop. The addition of these constants will allow us to extract programs that are more interesting than those obtainable from proofs in pure predicate logic, and which can be clearly compared to well-known stateful algorithms. To exemplify this, we will present a formally synthesised version of insertion sort, and we stress that by further extending our framework with additional rules and terms, we would be able to extract an even richer variety of combined functional/stateful programs.

%%%%%%%%%%%%%%%%%%%%%%%%%%%%%%%%%%%%%%%%%%%%%%%%%%%%%%%%%%%%%%%%%%%%%%%%%%%%%%%%%%%%%%%%%%
\subsection{The system $\ma$: First-order arithmetic with state}
\label{sec:arithmetic:logic}
%%%%%%%%%%%%%%%%%%%%%%%%%%%%%%%%%%%%%%%%%%%%%%%%%%%%%%%%%%%%%%%%%%%%%%%%%%%%%%%%%%%%%%%%%%

Our system of stateful intuitionistic arithmetic $\ma$ builds on $\ml$ just as ordinary first-order Heyting arithmetic builds on first-order predicate logic. In both cases, we introduce a constant $0$, a unary successor symbol $\suc$, symbols for all primitive recursive functions, and our predicate symbols now include an equality relation $=$. In what follows we write $x+1$ instead of $\suc(x)$. The axioms and rules of $\ma$ are, in turn, analogous to the additional axioms and rules we would require in ordinary first-order arithmetic: They include all axioms and rules of $\ml$ (based now on the language of $\ma$), along with a collection of additional axioms and rules. These comprise not only basic axioms and rules for equality and the successor, and an induction rule (all now adapted to incorporate the state), but also a new while rule for stateful iteration, which now exploits our state and, as we will see, allows us to extract programs that contain while loops. These additional axioms and rules are given in Figure 4.

\begin{figure}[ht]
\label{fig:ma}
\caption{Additional axioms and rules of $\ma$}
\begin{equation*}
\begin{gathered}
\textbf{Axioms and rules for equality}
\\[2mm]
\Gamma\mlpr \hr{\alpha}{t=t}{\alpha}
\qquad
%\vlinf{}{}{\Gamma\mlpr \hr{\alpha}{t=s}{\gamma}}{\Gamma\mlpr \hr{\alpha}{s=t}{\beta}}
\vlinf{}{}{\Gamma\mlpr \hr{\alpha}{t=s}{\beta}}{\Gamma\mlpr \hr{\alpha}{s=t}{\beta}} 
\qquad
\vlinf{}{}{\Gamma\mlpr \hr{\alpha}{r=t}{\gamma}} {\Gamma\mlpr \hr{\alpha}{r=s}{\beta} \ \ \ \Gamma\mlpr \hr{\beta}{s=t}{\gamma}}
%\\[2mm]
%\vlinf{}{}{\Gamma\mlpr \hr{\alpha_0}{f(s_1,\ldots,s_n)=f(t_1,\ldots,t_n)}{\alpha_{n}}}{\Gamma\mlpr \hr{\alpha_{i-1}}{s_i=t_i}{\alpha_{i}} \ \ \ i=1,\ldots,n}
%\\[2mm]
%\vlinf{}{ext}{\Gamma\mlpr \hr{\alpha_0}{P(t_1,\ldots,t_n)}{\beta(t_1,\ldots,t_n)}}{\Gamma\mlpr \hr{\alpha_{i-1}}{s_i=t_i}{\alpha_{i}} \ \ \ i=1,\ldots,n \ \ \ \Gamma \mlpr \hr{\alpha_n}{P(s_1,\ldots,s_n)}{\beta(s_1,\ldots,s_n)}}
\\[2mm]
\vlinf{}{ext}{\Gamma\mlpr \hr{\alpha}{A(t)}{\gamma(t)}}{\Gamma\mlpr \hr{\alpha}{s=t}{\beta} \ \ \ \Gamma \mlpr \hr{\beta}{A(s)}{\gamma(s)}}
\\[2mm]
\textbf{Axioms and rules for arithmetical function symbols}
\\[2mm]
\Gamma\mlpr\hr{\alpha}{\suc(t)\neq 0}{\alpha}
\qquad
\vlinf{}{}{\hr{\alpha}{s=t}{\beta}}{\hr{\alpha}{\suc(s)=\suc(t)}{\beta}}
\\[2mm]
\Gamma\mlpr\hr{\alpha}{l=r}{\alpha} \ \ \ \mbox{where $l=r$ ranges across defining equations for prim. rec. functions}
\\[2mm]
\textbf{Induction rule}
\\[2mm]
\vliinf{}{ind}{\Gamma\mlpr\hr{\gamma}{\forall x\,\hr{\alpha}{A(x)}{\beta(x)}}{\gamma}}{\Gamma\mlpr \hr{\alpha}{A(0)}{\beta(0)}}{\Gamma,A(x)\mlpr\hr{\beta(x)}{A(x+1)}{\beta(x+1)}}
\\[2mm]
\textbf{While rule (over natural numbers)}
\\[2mm]
\vliiinf{}{while}{\Gamma,A(x)\mlpr \hr{\alpha(x)}{B}{\beta}}{\mathcal{A}_1}{\mathcal{A}_2}{\mathcal{A}_3}
\\[2mm]
\mathcal{A}_1:=\Gamma,A(x+1)\mlpr\hr{\gamma(x+1)\wedge\alpha(x+1)}{A(x)}{\alpha(x)}
\\
\mathcal{A}_2:=\Gamma,A(x+1)\mlpr\hr{\neg\gamma(x+1)\wedge \alpha(x+1)}{B}{\beta}
\\
\mathcal{A}_3:=\Gamma,A(0)\mlpr\hr{\alpha(0)}{B}{\beta}
\\[2mm]
\mbox{for $ind$ and $while$, $x$ is not free in $\Gamma$, and for $while$ it is not free in $B$ or $\beta$}
\end{gathered}
\end{equation*}
\end{figure} 
Our formulation of stateful arithmetic follows the same basic idea as the construction of stateful predicate logic, incorporating standard rules but keeping track of an ambient state in a call-by-value manner, and adding new rules that explicitly correspond to stateful constructions. In particular, Proposition \ref{stateful:extension} clearly extends to $\ma$, as the usual axioms and rules of arithmetic can be embedded into those of $\ma$:
\begin{prop}
\label{stateful:extension:arithmetic}
For any formula $A$ of $\ha$ and state formula $\alpha$, define the main formula $A_\alpha$ of $\ma$ as in Proposition \ref{stateful:extension}. Then whenever $\Gamma\plpr A$ is provable in $\ha$, we have that $\Gamma_\alpha,\Delta\mlpr\hr{\alpha}{A_\alpha}{\alpha}$ is provable in $\ma$, where $\Delta$ is arbitrary and $\Gamma_\alpha:=(A_1)^{u_1}_\alpha,\ldots,(A_n)^{u_n}_\alpha$ for $\Gamma:=A_1^{u_1},\ldots,A_n^{u_n}$.
\end{prop}
We can also derive a natural extensionality rule from our stateful equality rules, which assures us that whenever $s=t$ in ordinary Heyting arithmetic, then we can replace $s$ by $t$ for stateful formulas:
\begin{prop}
\label{stateful:extensionality}
Suppose that $\plpr s=t$ is provable in $\ha$. Then from $\Gamma\mlpr\hr{\alpha(s)}{A(s)}{\beta(s)}$ we can derive $\Gamma\mlpr\hr{\alpha(t)}{A(t)}{\beta(t)}$ in $\ma$.
\end{prop}

\begin{proof}
By Proposition \ref{stateful:extension:arithmetic} for $\alpha:=\alpha(s)$ we have $\Gamma\mlpr\hr{\alpha(s)}{s=t}{\alpha(s)}$ and thus using the extensionality rule in $\ma$ we can derive
\begin{equation*}
\vlinf{}{ext}{\Gamma\mlpr\hr{\alpha(s)}{A(t)}{\beta(t)}}{\Gamma\mlpr \hr{\alpha(s)}{s=t}{\alpha(s)} \ \ \  \Gamma \mlpr \hr{\alpha(s)}{A(s)}{\beta(s)}}
\end{equation*}
Since $\plpr t=s$ must also be provable in $\ha$, another instance of Proposition \ref{stateful:extension:arithmetic} for $\alpha:=\alpha(t)$ along with the true axiom in $\ma$ gives us
\begin{equation*}
\vlinf{}{ext}{\Gamma\mlpr \hr{\alpha(t)}{\top}{\alpha(s)}}{\Gamma\mlpr \hr{\alpha(t)}{t=s}{\alpha(t)} \ \ \ \Gamma \mlpr \hr{\alpha(t)}{\top}{\alpha(t)}}
\end{equation*}
Putting these together we obtain
\begin{equation*}
\vlderivation{
\vlin{}{\wedge_SE_L}{\Gamma\mlpr\hr{\alpha(t)}{A(t)}{\beta(t)}}
	{\vliin{}{\wedge_S I}{\Gamma\mlpr\hr{\alpha(t)}{\top\wedge A(t)}{\beta(t)}}
	{\vlhy{\Gamma\mlpr \hr{\alpha(t)}{\top}{\alpha(s)}}}
	{\vlhy{\Gamma\mlpr\hr{\alpha(s)}{A(t)}{\beta(t)}}}	
	}
}
\end{equation*}
which completes the derivation.
\end{proof}

%%%%%%%%%%%%%%%%%%%%%%%%%%%%%%%%%%%%%%%%%%%%%%%%%%%%%%%%%%%%%%%%%%%%%%%%%%%%%%%%%%%%%%%%%%
\subsection{An extended term calculus $\mta$}
\label{sec:arithmetic:term}
%%%%%%%%%%%%%%%%%%%%%%%%%%%%%%%%%%%%%%%%%%%%%%%%%%%%%%%%%%%%%%%%%%%%%%%%%%%%%%%%%%%%%%%%%%

In order to give derivations in $\ma$ a computation interpretation, we need to extend our term calculus $\mt$ to include a recursor (for induction) and a controlled while loop (for the while rule). The remaining new axioms and rules of $\ma$ are dealt with in a straightforward manner.

To be precise: the theory $\mta$ is defined to be the instance of $\mt$ for the case of arithmetic, with function symbols for zero, successor and all primitive recursive functions. Accordingly, we rename the base type $D$ to $\nat$. In addition to the terms of $\mt$, we add terms $\rec{e}{e}$ and $\while{\gamma[z]}{e}{e}{e}{e}$ to our grammar, where $\gamma[z]$ ranges over state formulas of $\ml$ with a specified free variable $z$. The typing rules for these new terms are
\begin{equation*}
\frac{\Gamma\vdash s:X \ \ \ \Gamma\vdash t:\nat \to X \to X}{\Gamma\vdash\rec{s}{t}:\nat\to X}
\end{equation*}
for the recursor, while for the while loop we have
\begin{equation*}
\frac{\Gamma\vdash r:\nat\to X\to X \ \ \ \Gamma\vdash s:\nat\to X\to Y \ \ \ \Gamma\vdash t:X\to Y \ \ \ \Gamma\vdash u:\nat}{\Gamma\vdash\while{\gamma[z]}{r}{s}{t}{u}:X\to Y}
\end{equation*}
under the additional variable condition that $z\notin\Gamma$, but $x:\nat\in\Gamma$ for all free variables of $\gamma[z]$ outside of $z$. Note that we do not consider $z$ a free variable of $\while{\gamma[z]}{r}{s}{t}{a}$, but rather a placeholder for the loop condition. In order to give the appropriate semantics to our terms, we must add to our metatheory $\meta$ axioms and rules for arithmetic in all finite types, including the ability to define functions of arbitrary type via recursion over the natural numbers, along the lines of $\mathrm{E}\mbox{-}\mathrm{HA}^\omega$ \cite{troelstra:73:book} (though as before the precise details are not important). We then define:
\begin{itemize}

	\item $[\rec{s}{t}]\pi:=\pair{R_f,\pi_1}$ for $\pair{f,\pi_1}:=[t]\pi$, where
	\begin{equation}
	\label{def:R}
	\begin{aligned}
	R_f0\pi&:=[s]\pi\\
	R_f(n+1)\pi&:=ga\pi'_2 \mbox{ for $\pair{a,\pi'_1}:=R_fn\pi'$ and $\pair{g,\pi'_2}:=fn\pi'_1$}
	\end{aligned}
	\end{equation}
	
	\item $[\while{\gamma[z]}{r}{s}{t}{u}]\pi:=\pair{L_{f,g,h}m,\pi_4}$ where $\pair{f,\pi_1}:=[r]\pi$, $\pair{g,\pi_2}:=[s]\pi_1$, $\pair{h,\pi_3}:=[t]\pi_2$ and $\pair{m,\pi_4}:=[u]\pi_3$, where
	\begin{equation}
	\label{def:L}
	\begin{aligned}
	&L_{f,g,h}0y\pi':=hy\pi'\\
	&L_{f,g,h}(n+1)y\pi'\\ &:=\begin{cases}
		L_{f,g,h}ny'\pi_2 \mbox{ for $\pair{a,\pi'_1}:=fn\pi'$ and $\pair{y',\pi'_2}:=ay\pi'_1$} & \mbox{if $[\gamma][n+1](\pi')$}\\
		by\pi'_1 \mbox{ for $\pair{b,\pi'_1}:=gn\pi'$} & \mbox{if $\neg[\gamma][n+1](\pi')$}
	\end{cases}
	\end{aligned}
	\end{equation}
where in the case distinctions, we would technically speaking need to use the characteristic function $\chi_{\gamma}\pair{x_1,\ldots,n,\ldots, x_k}$ for $\gamma$, with $n$ substituted for the special free variable $z$.
\end{itemize}

%%%%%%%%%%%%%%%%%%%%%%%%%%%%%%%%%%%%%%%%%%%%%%%%%%%%%%%%%%%%%%%%%%%%%%%%%%%%%%%%%%%%%%%%%%
\subsection{The soundness theorem for arithmetic}
\label{sec:arithmetic:soundness}
%%%%%%%%%%%%%%%%%%%%%%%%%%%%%%%%%%%%%%%%%%%%%%%%%%%%%%%%%%%%%%%%%%%%%%%%%%%%%%%%%%%%%%%%%%

We now need to show that the soundness proof for stateful predicate logic also holds in the extension to arithmetic.
\begin{thm}
\label{thm:main:arithmetic}
The statement of Theorem \ref{thm:main} remains valid if we replace $\ml$ by $\ma$ and $\mt$ by $\mta$.
\end{thm}

\begin{proof}
We need to extend the proof of Theorem \ref{thm:main:arithmetic} to show that the additional axioms and rules as in Figure \ref{fig:ma} can be realized by a term of the form $[t]$ for $t$ in $\mta$.\medskip
\begin{itemize}

	\item For the non-extensionality equality and arithmetic axioms this is straightforward due to the fact that these are also true in $\meta$: For instance, given a realizer $\sr{[s]}{\hr{\alpha}{u=v}{\beta}}$ and $\sr{[t]}{\hr{\beta}{v=w}{\gamma}}$, we have that $\sr{[s\circ t]}{\hr{\alpha}{u=v\wedge v=w}{\gamma}}$, and since from $u=v\wedge v=w$ we can infer $u=w$ in $\meta$, it follows that $\sr{[p_1(s\circ t)]}{\hr{\alpha}{u=w}{\gamma}}$. The other axioms and rules are even simpler.\smallskip
	
	\item $(ext)$ Extensionality is similarly simple: If $\sr{[s]}{\hr{\alpha}{u=v}{\beta}}$ and $\sr{[t]}{\hr{\beta}{A(u)}{\gamma(u)}}$, then $[\alpha](\pi)$ implies that $u=v$ and $[\beta](\pi_1)$ for $\pair{\ldots,\pi_1}:=[s]\pi$, and therefore $\sr{a}{A(u)}$ and $[\gamma](u)(\pi_2)$ for $\pair{a,\pi_2}:=[t]\pi_1$. Now applying extensionality in $\meta$ to the formula $T(x):=\sr{a}{A(x)}\wedge [\gamma](x)(\pi_2)$, from $u=v$ we have $\sr{a}{A(v)}$ and $[\gamma](v)(\pi_2)$, and thus $\sr{[s\circ t]}{\hr{\alpha}{u=v\wedge A(v)}{\gamma(v)}}$ and therefore $\sr{[p_2(s\circ t)]}{\hr{\alpha}{A(v)}{\gamma(v)}}$.\smallskip
	
	\item $(rec)$ Suppose that $s$ and $t(x,y)$ are such that $\sr{[s]}{\hr{\alpha}{A(0)}{\beta(0)}}$ and $$\sr{[t](x,y)}{\hr{\beta(x)}{A(x+1)}{\beta(x+1)}}$$ assuming $\sr{y}{A(x)}$. We show that $\sr{[\rec{s}{\lambda x,y.t(x,y)}]}{\hr{\gamma}{\forall x\, \hr{\alpha}{A(x)}{\beta(x)}}{\gamma}}$ for any $\gamma$. Since $[\rec{s}{\lambda x,y.t(x,y)}]\pi=\pair{R_f,\pi}$ for $f:=\lambda x.[\lambda y.t(x,y)]$ and $R_f$ as in (\ref{def:R}), it suffices to show that for any $n:\nat$ we have 
	\begin{equation*}
	\sr{R_fn}{\hr{\alpha}{A(n)}{\beta(n)}}
	\end{equation*} 
	We prove this by induction: For the base case, we have $R_f0=[s]$ and the claim holds by assumption. For the induction step, let us assume that $[\alpha](\pi')$ holds, and so by the induction hypothesis we have $\sr{a}{A(n)}$ and $[\beta(n)](\pi'_1)$ for $\pair{a,\pi'_1}:=R_fn$. Since $fn\pi'_1=\pair{g,\pi'_1}$ for $g:=\lambda y.[t](n,y)$, we have that $R_f(n+1)\pi'=[t](n,a)\pi'_1$, and since by the property of $[t]$ we then have $\sr{b}{A(n+1)}$ and $[\beta(n+1)](\pi'_2)$ for $\pair{b,\pi'_2}:=[t](n,a)\pi'_1$, we have shown that $\sr{R_f(n+1)}{\hr{\alpha}{A(n+1)}{\beta(n+1)}}$, which completes the induction.\smallskip
	
	\item $(while)$ We suppose that \smallskip
	\begin{enumerate}
	
	\item $\sr{[r](x,y)}{\hr{\gamma(x+1)\wedge\alpha(x+1)}{A(x)}{\alpha(x)}}$ assuming that $\sr{y}{A(x+1)}$, \smallskip
	
	\item $\sr{[s](x,y)}{\hr{\neg\gamma(x+1)\wedge\alpha(x+1)}{B}{\beta}}$ assuming that $\sr{y}{A(x+1)}$, \smallskip
	
	\item $\sr{[t](y)}{\hr{\alpha(0)}{B}{\beta}}$ assuming that $\sr{y}{A(0)}$. \smallskip
	
	\end{enumerate}
	Our aim is to show that
	\begin{equation*}
	\sr{[(\while{\gamma}{(\lambda x',y'.r)}{(\lambda x',y'.s)}{(\lambda y'.t)}{x})y]}{\hr{\alpha(x)}{B}{\beta}}	
	\end{equation*}
	for any $x,y\in\nat$ with $\sr{y}{A(x)}$. We observe, unwinding the definition, that
	\begin{equation*}
	[(\while{\gamma}{(\lambda x',y'.r)}{(\lambda x',y'.s)}{(\lambda y'.t)}{x})y]\pi=L_{f,g,h}xy\pi
	\end{equation*}
	for $f:=\lambda x'.[\lambda y'.r(x',y')]$, $g:=\lambda x'.[\lambda y'.s(x',y')]$, $h:=\lambda y'.[t](y')$ and $L_{f,g,h}$ as defined in (\ref{def:L}). We now show by induction on $n$ that if $\sr{y}{A(n)}$ then
	\begin{equation*}
	\sr{L_{f,g,h}ny}{\hr{\alpha(n)}{B}{\beta}}
	\end{equation*}
	and then the result follows by setting $n:=x$. The base case is straightforward since 
	\begin{equation*}
	L_{f,g,y}0y=[t](y)
	\end{equation*}
	and the claim follows by definition of $[t]$. For the induction step, suppose that $\sr{y}{A(n+1)}$ and $[\alpha(n+1)](\pi)$. There are two cases. If $\neg[\gamma](n+1)(\pi)$ we have
	\begin{equation*}
	L_{f,g,h}(n+1)y\pi=[s](n,y)\pi
	\end{equation*}
	and the result holds by the property of $[s]$. On the other hand, if $[\gamma](n+1)(\pi)$ then
	\begin{equation*}
	L_{f,g,h}(n+1)y\pi=L_{f,g,h}ny'\pi'
	\end{equation*}
	for $\pair{y',\pi'}:=[r](n,y)\pi$. But by the property of $[r]$ we have $\sr{y'}{A(n)}$ and $[\alpha(n)](\pi')$, and therefore by the induction hypothesis we have $\sr{b}{B}$ and $[\beta](\pi'')$ for $\pair{b,\pi''}:=L_{f,g,h}ny'\pi'=L_{f,g,h}(n+1)y\pi$, and so the result is proven for $n+1$.

\end{itemize}\medskip
This covers all the additional axioms and rules of $\ma$.
\end{proof}

%%%%%%%%%%%%%%%%%%%%%%%%%%%%%%%%%%%%%%%%%%%%%%%%%%%%%%%%%%%%%%%%%%%%%%%%%%%%%%%%%%%%%%%%%%
\subsection{Worked example: Insertion sort}
\label{sec:arithmetic:example}
%%%%%%%%%%%%%%%%%%%%%%%%%%%%%%%%%%%%%%%%%%%%%%%%%%%%%%%%%%%%%%%%%%%%%%%%%%%%%%%%%%%%%%%%%%

We now illustrate our extended system by synthesising a list sorting program that, intuitively, forms an implementation of the insertion sort algorithm. Here our state will represent the structure that is to be sorted, and continuing the spirit of generality that we have adhered to throughout, we characterise this structure through a number of abstract axioms. Instantiating the state as, say, an array of natural numbers, would provide a model for our theory, but our sorting algorithm can be extracted on the more abstract level. Crucially, the proof involves both loop iteration and induction, and the corresponding program combines an imperative while loop with a functional recursor.

We begin by axiomatising our state, just as in previous examples. An intuition here is that states represent an infinite array of elements $a_0,a_1,\ldots$ possessing some total order structure $\leq$, and we seek to extract a program that, for any input $n$, sorts the first $n$ elements. We use this informal semantics throughout to indicate the intended meaning of our axioms, but stress that none of this plays a formal role in the proof or resulting computational interpretation.

We introduce three state predicates to $\ma$, with the intuition indicated in each case:\medskip
\begin{itemize}

	\item $\sort(N)$ -- \emph{Sorted: The first $N+1$ elements of the array i.e. $[a_0,\ldots,a_N]$ are sorted}\smallskip
	
	\item $\psort(n,N)$ -- \emph{Partially sorted with respect to $a_n$: if $n<N$ then the list $$[a_0,\ldots,a_{n-1},a_{n+1},\ldots,a_N]$$ is sorted and $a_n\leq a_{n+1}$. For the base cases, if $n=N$ then the list $[a_0,\ldots,a_{N-1}]$ is sorted, and if $n>N$ then the list $[a_0,\ldots,a_N]$ is sorted.}\smallskip
	
	\item $\compare(n)$ -- \emph{Comparison: true if $a_n\leq a_{n-1}$, and always true if $n=0$}\medskip

\end{itemize}
We formalise this intuition by adding the following state independent axioms to $\axh$:\medskip
\begin{enumerate}

	\item $\Gamma,\sort(N)\hlpr\psort(N+1,N+1)$ -- \emph{If the first $N+1$ elements are sorted, then they are also partially sorted with respect to the next element $a_{N+1}$.}\smallskip
	
	\item $\Gamma,\neg\compare(n),\psort(n,N)\hlpr \sort(N)$ -- \emph{If $[a_0,\ldots,a_{n-1},a_{n+1},\ldots,a_N]$ is sorted, $a_n\leq a_{n+1}$, but also $a_{n-1}\leq a_n$, then the entire segment $[a_0,\ldots,a_N]$ must be sorted.}\smallskip
	
	\item $\Gamma,\psort(0,N)\hlpr \sort(N)$ -- \emph{If $[a_1,\ldots,a_N]$ is sorted and $a_0\leq a_1$, then $[a_0,\ldots,a_N]$ is sorted.}\smallskip
	
	\item $\Gamma\hlpr\sort(0)$ -- \emph{The singleton array $[a_0]$ is defined to be sorted.}\medskip

\end{enumerate}
We complete the axiomatisation by adding a single state-sensitive axiom to $\axm$:\medskip
\begin{enumerate}

	\item[(5)] $\Gamma\mlpr\hr{\compare(n+1)\wedge\psort(n+1,N)}{\top}{\psort(n,N)}$ -- \emph{If $[a_0,\ldots,a_n,a_{n+2},\ldots,a_N]$ is sorted and $a_{n+1}\leq a_{n+2}$, but $a_{n+1}\leq a_{n}$, then we can modify the state (i.e. swapping $a_n$ and $a_{n+1}$ by setting $\tilde{a}_n:=a_{n+1}$ and $\tilde{a}_{n+1}:=a_n$) so that $[a_0,\ldots,a_{n-1},\tilde{a}_{n+1},\ldots,a_N]$ is sorted and $\tilde{a}_n\leq \tilde{a}_{n+1}$. The edge cases for $n\geq N$ are interpreted in a more straightforward way.}	\medskip

\end{enumerate}
In order to give a realizing term to this axiom, we representing element swapping semantically by adding a constant $c:\nat\to S\to S$ to our metatheory $\meta$, which satisfies
\begin{equation*}
\compare(n+1,\pi)\wedge\psort(n+1,N,\pi)\implies \psort(n,N,cn\pi)
\end{equation*}
and a corresponding term $\swap:\nat\to C$ to our term calculus, along with the embedding 
\begin{equation*}
[\swap]\pi:=\pair{\lambda n,\pi.\pair{(),cn\pi},\pi}\simeq\pair{c,\pi}
\end{equation*}
so that we can prove
\begin{equation*}
\sr{[\swap\, n]}{\hr{\compare(n+1)\wedge\psort(n+1,N)}{\top}{\psort(n,N)}}
\end{equation*}
With this in place, we can now prove in $\ma$ that the first $N$ elements of the state can be sorted, and extract a corresponding realizing term in $\mta$.
%
%%%%%%%%%%%%%%%%%%%%%%%%%%%%%%%%%%%%%%%%%%%%%%%%%%%%%%%%%%%%%%%%%%%%%%%%%%%%%%%%%%%%%%%%%%
\subsubsection{Proof of $\mlpr\hr{\gamma}{\forall N\, \hr{\alpha}{\top}{\sort(N)}}{\gamma}$ in $\ma$}
\label{sec:arithmetic:example:proof}
%%%%%%%%%%%%%%%%%%%%%%%%%%%%%%%%%%%%%%%%%%%%%%%%%%%%%%%%%%%%%%%%%%%%%%%%%%%%%%%%%%%%%%%%%%

The core of our proof begins with an instance of the $while$ rule parametrised by $N$, with $\Gamma:=\emptyset$, $A(n):=\top$, $\alpha(n):=\psort(n,N+1)$, $\beta:=\sort(N+1)$ and $\gamma(n):=\compare(n)$:
\begin{equation*}
\vlderivation{
\vlin{}{cons}
	{\top\mlpr\hr{\sort(N)}{\top}{\sort(N+1)}}
	{
	\vlin{}{\forall_SE}
		{\top\mlpr\hr{\psort(N+1,N+1)}{\top}{\sort(N+1)}}
		{
		\vlin{}{\forall_SI}
			{\top\mlpr\hr{\psort(N+1,N+1)}{\forall n\, \hr{\psort(n,N+1)}{\!\top\!}{\sort(N+1)}}{\psort(N+1,N+1)}}
			{
			\vlin{}{while}
				{\top\mlpr\hr{\psort(n,N+1)}{\top}{\sort(N+1)}}
				{
				\vlhy{\mathcal{D}_1 \ \ \ \mathcal{D}_2 \ \ \ \mathcal{D}_3}
				}
			}
		}
	}
}
\end{equation*}
where the final composition inference makes use of the first state independent axiom. Here $\mathcal{D}_1$ represents an instance of the state sensitive axiom
\begin{equation*}
\top\mlpr\hr{\compare(n+1)\wedge\psort(n+1,N+1)}{\top}{\psort(n,N+1)}
\end{equation*}
and $\mathcal{D}_2$ represents the derivation
\begin{equation*}
\vlderivation{
	\vlin{}{cons}
		{\top\mlpr\hr{\neg\compare(n+1)\wedge\psort(n+1,N+1)}{\top}{\sort(N+1)}}
		{\vlhy{\top\mlpr\hr{\sort(N+1)}{\top}{\sort(N+1)}}}
}
\end{equation*}
where composition makes use of the second state independent axiom. Finally $\mathcal{D}_3$ is
\begin{equation*}
\vlderivation{
	\vlin{}{cons}
		{\top\mlpr\hr{\psort(0,N+1)}{\top}{\sort(N+1)}}
		{\vlhy{\top\mlpr\hr{\sort(N+1)}{\top}{\sort(N+1)}}}
}
\end{equation*}
this time making use of the third state independent axiom. Finally we can prove that all lists can be sorted with an outer induction as follows:
\begin{equation*}
\vlderivation{
	\vliin{}{ind}
		{\mlpr\hr{\gamma}{\forall N\, \hr{\alpha}{\top}{\sort(N)}}{\gamma}}
		{\vlin{}{cons}{\mlpr\hr{\alpha}{\top}{\sort(0)}}{\vlhy{\mlpr\hr{\alpha}{\top}{\alpha}}}}
		{\vlin{}{}{\top\mlpr\hr{\sort(N)}{\top}{\sort(N+1)}}{\vlhy{\mathcal{D}}}}
}
\end{equation*}
where $\alpha$ is an arbitrary state predicate, the instance of $cons$ uses the fourth state independent axiom, and $\mathcal{D}$ represents the derivation above.

%%%%%%%%%%%%%%%%%%%%%%%%%%%%%%%%%%%%%%%%%%%%%%%%%%%%%%%%%%%%%%%%%%%%%%%%%%%%%%%%%%%%%%%%%%
\subsubsection{Program extraction}
\label{sec:arithmetic:example:program}
%%%%%%%%%%%%%%%%%%%%%%%%%%%%%%%%%%%%%%%%%%%%%%%%%%%%%%%%%%%%%%%%%%%%%%%%%%%%%%%%%%%%%%%%%%

We now extract a program that corresponds to the above proof. First of all, we note that the three premises of our while rule are realised by $\swap\, n$, $\unit$ and $\unit$ respectively, and so our derivation $\mathcal{D}$ corresponds to the following program:
\begin{equation*}
\vlderivation{
\vlin{}{cons}
	{y:C\vdash(\lambda n.t(n)y)(N+1):C}
	{
	\vlin{}{\forall_SE}
		{y:C\vdash(\lambda n.t(n)y)(N+1):C}
		{
		\vlin{}{\forall_SI}
			{y:C\vdash\lambda n.t(n)y:\nat\to C}
			{
			\vlin{}{while}
				{y:C\vdash t(n)y:C}
				{
				\vlhy{y:C\vdash \swap\, n : C \ \ \ y:C\vdash \unit:C \ \ \ y:C\vdash \unit:C}
				}
			}
		}
	}
}
\end{equation*}
where
\begin{equation*}
\begin{aligned}
t(n)&:=\while{\compare[z]}{(\lambda x,y.(\swap\, x))}{(\lambda x,y. \unit)}{(\lambda y.\unit)}{n}\\
&\simeq \while{\compare[z]}{(\lambda x.(\swap\, x))}{(\unit)}{(\unit)}{n}
\end{aligned}
\end{equation*}
Then our final induction generates the following program:
\begin{equation*}
\vlderivation{
	\vliin{}{ind}
		{\vdash\rec{(\unit)}{(\lambda x,y.((\lambda n.t(n)y)(x+1)))}:\nat\to C}
		{\vlhy{\vdash\unit:C}}
		{\vlhy{y:C\vdash(\lambda n.t(n)y)(N+1):C}}
}
\end{equation*}
Thus our list sorting program is
\begin{equation*}
\begin{aligned}
&\rec{(\unit)}{(\lambda x,y.((\lambda n.t(n)y)(x+1)))}\\
&\simeq \rec{(\unit)}{(\lambda x.((\lambda n.(\while{\compare[z]}{(\lambda x.(\swap\, x))}{(\unit)}{(\unit)}{n}()))(x+1)))}
\end{aligned}
\end{equation*}
which is essentially an implementation of the insertion sort algorithm, with an outer recursion that sorts initial segments of the list in turn, and an inner loop that inserts new elements into the appropriate place in the current sorted list.

%%%%%%%%%%%%%%%%%%%%%%%%%%%%%%%%%%%%%%%%%%%%%%%%%%%%%%%%%%%%%%%%%%%%%%%%%%%%%%%%%%%%%%%%%%
%%%%%%%%%%%%%%%%%%%%%%%%%%%%%%%%%%%%%%%%%%%%%%%%%%%%%%%%%%%%%%%%%%%%%%%%%%%%%%%%%%%%%%%%%%
\section{Directions for future work}
\label{sec:future}
%%%%%%%%%%%%%%%%%%%%%%%%%%%%%%%%%%%%%%%%%%%%%%%%%%%%%%%%%%%%%%%%%%%%%%%%%%%%%%%%%%%%%%%%%%
%%%%%%%%%%%%%%%%%%%%%%%%%%%%%%%%%%%%%%%%%%%%%%%%%%%%%%%%%%%%%%%%%%%%%%%%%%%%%%%%%%%%%%%%%%

In this paper we have presented the central ideas behind a new method for extracting stateful programs from proofs, which include an extension of ordinary first-order logic with Hoare triples, a corresponding realizability interpretation, and a soundness theorem. We emphasise once again that our intention has been to offer an alternative approach to connecting proofs with stateful programs, one that seeks to complement rather than improve existing work by embracing simplicity and abstraction, and which might be well suited to a range of applications in proof theory or computability theory. In this spirit, we conclude with a very informal outline of a series interesting directions in which we anticipate that our framework could be applied.

%%%%%%%%%%%%%%%%%%%%%%%%%%%%%%%%%%%%%%%%%%%%%%%%%%%%%%%%%%%%%%%%%%%%%%%%%%%%%%%%%%%%%%%%%%
\subsection{Further extensions and program synthesis}
\label{sec:future:arithmetic}
%%%%%%%%%%%%%%%%%%%%%%%%%%%%%%%%%%%%%%%%%%%%%%%%%%%%%%%%%%%%%%%%%%%%%%%%%%%%%%%%%%%%%%%%%%  

While our main results have been presented in the neutral setting of first-order predicate logic, it would be straightforward to extend $\ml$ to richer logics with more complex data structures and imperative commands.  Already, the addition of recursion and loops over natural numbers in Section \ref{sec:arithmetic} has allowed us to synthesise a standard in-place sorting algorithm using our abstract axiomatisation of an ordered state, in a similar spirit to \cite{berger-etal:14:imperative}. However, further extensions are naturally possible, including the addition of general fixpoint operators and non-controlled while loops, which would then require a $\meta$ to be replaced by a domain theoretic semantics that allows for partiality. 

Looking a step further ahead, by implementing all of this in a proof assistant, we would have at our disposal a new technique for synthesising correct-by-construction imperative programs. While we do not suggest that this pipeline would directly compete with existing techniques for verifying imperative programs, it could be well suited to synthesising and reasoning about programs in very specific domains, where we are interested in algorithms for which interactions with the state have a restricted form that could be suitably axiomatised within our logic. For example, a more detailed axiomatisation our state as an ordered array along the lines of Section \ref{sec:arithmetic:example}, with a ``swap'' operation and a few other ways of interacting with the state, might give rise to an interesting theory of in-place sort algorithms. Stateful algorithms on other data structures, such as graphs, could presumably also be formalised within our framework.

%%%%%%%%%%%%%%%%%%%%%%%%%%%%%%%%%%%%%%%%%%%%%%%%%%%%%%%%%%%%%%%%%%%%%%%%%%%%%%%%%%%%%%%%%%
\subsection{Bar recursion and the semantics of extracted programs}
\label{sec:future:semantics}
%%%%%%%%%%%%%%%%%%%%%%%%%%%%%%%%%%%%%%%%%%%%%%%%%%%%%%%%%%%%%%%%%%%%%%%%%%%%%%%%%%%%%%%%%% 

Two of the main starting points for this paper, the monadic realizability of Birolo \cite{birolo:12:phd} and the author's own Dialectica interpretation with state \cite{powell:18:state}, address the broader problem of trying to understand the operational semantics of programs extracted from proofs as stateful procedures (the origins and development of this general idea, from Hilbert's epsilon calculus onwards, is brilliantly elucidated in Chapter 1 of Aschieri's thesis \cite{aschieri:11:phd}, who then sets out his own realizability interpretation based on learning). A number of case studies by the author and others \cite{oliva-powell:15:ramsey,oliva-powell:15:bolzano,powell:20:quasiorder,powell-etal:22:krull} have demonstrated that while terms extracted from nontrivial proofs can be extremely complex, they are often much easier to understand if one focuses on the way they interact with the mathematical environment. For example, in understanding a program extracted from a proof using Ramsey's theorem for pairs \cite{oliva-powell:15:ramsey}, it could be illuminating to study the \emph{trace} of the program as it queries a colouring at particular pairs, as this can lead to a simpler characterisation of the \emph{algorithm} ultimately being implemented by the term. 

While the aforementioned analysis of programs has always been done in an informal way, our stateful realizability interpretation would in theory allow us to extract programs which store this trace formally in the state, where our abstract characterisation of state would allow us to implement it in whichever way is helpful in a given setting. For example, in the case of the Bolzano-Weierstrass theorem \cite{oliva-powell:15:bolzano}, our state might record information of the form $x_n\in I$, collecting information about the location of sequence elements. For applications in algebra \cite{powell-etal:22:krull}, one might instead store information about a particular maximal ideal.  
The aforementioned theorems are typically proven using some form of choice or comprehension, and that in itself leads to the interesting prospect of introducing both stateful recursors and while-loops that are computationally equivalent to variants of \emph{bar recursion} \cite{spector:62:barrecursion}. In \cite{powell:16:learning}, several bar recursive programs that arise from giving a computational interpretation to arithmetical comprehension principles are formulated as simple while loops, and these could in principle be incorporated into our system with new controlled Hoare rules in the style of update recursion \cite{berger:04:open}, that replace the conditions $n<N$ and $n\geq N$ in the $\mathcal{A}_i$ above with e.g. $n\in\mathrm{dom}(f)$ and $n\notin  \mathrm{dom}(f)$, where $f$ is some partial approximation to a comprehension function. An exploration of such while-loops from the perspective of higher-order computability theory might well be of interest in its own right.

%%%%%%%%%%%%%%%%%%%%%%%%%%%%%%%%%%%%%%%%%%%%%%%%%%%%%%%%%%%%%%%%%%%%%%%%%%%%%%%%%%%%%%%%%%
\subsection{A logic for probabilistic lambda calculi}
\label{sec:future:probabilistic}
%%%%%%%%%%%%%%%%%%%%%%%%%%%%%%%%%%%%%%%%%%%%%%%%%%%%%%%%%%%%%%%%%%%%%%%%%%%%%%%%%%%%%%%%%%

Probabilistic functional languages are a major topic of research at present. While work in this direction dates back to the late 1970s \cite{jones-plotkin:89:probabilistic,sahebdjaromi:78:probabilistic} where it typically had a semantic flavour, a more recent theme \cite{dallago-zorzi:08:probabilistic,deliguoro-piperno:95:nondeterministic,dipierro-etal:05:probabilistic} has been to study simple extensions of the lambda calculus with nondeterministic choice operators $\oplus$, where $s\oplus t$ evaluates nondeterministically (or probabilistically) to either $s$ or $t$. While such calculi have been extensively studied, corresponding \emph{logics} that map under some proof interpretation to probabilistic programs are far more rare (although there is some recent work in this direction e.g. \cite{antonelli-etal:22:borel}).

We conjecture that our framework offers a bridge between logic and probabilistic computation through incorporating probabilistic disjunctions into our logic $\ml$ and taking states to be streams of outcomes of probabilistic events together with a current `counter' that increases each time an event occurs. In a simple setting where only two outcomes are possible with equal probability, we can axiomatise this within $\ml$ by adding zero and successor functions (allowing us to create numerals $n$), along with a unary state predicate $\scount(n)$. We can then model probabilistic events by adding the appropriate axioms to $\axm$. Suppose, for example, we add two predicate constants $H(x)$ and $T(x)$ (for \emph{heads} and \emph{tails}), along with constants $c_1,c_2,\ldots$ representing coins. Then flipping a coin would be represented by the axiom schema
\begin{equation*}
\Gamma\mlpr\hr{\scount(n)}{H(c_i)\vee T(c_i)}{\scount(n+1)}
\end{equation*}
where $n$ ranges over numerals and $c_i$ over coin constants, the counter indicating that a probabilistic event has occurred. The act of reading a probability from the state could be interpreted semantically by introducing a constant $\omega:S\to \bool\times S$ to $\meta$, with the axiom 
\begin{equation*}
\scount(n,\pi)\implies (e=\false\implies H(c_i))\wedge (e=\true\implies T(c_i))\wedge\scount(n+1,\pi_1) \mbox{ for $\pair{e,\pi_1}:=\omega\pi$}
\end{equation*}
(alternatively, we could simply define $S:=\nat\times(\nat\to\bool)$ for a type of $\nat$ natural numbers, and define $\omega\pair{n,a}:=\pair{a(n),\pair{n+1,a}}$ and $\scount(n,\pair{m,a}):=m=_\nat n$).

A probabilistic choice operator $\oplus$ can then be added to the language of $\mt$, along with the typing rule $\Gamma\vdash s\oplus t:X+Y$ for $\Gamma\vdash s:X$ and $\Gamma\vdash t:Y$, and the interpretation
\begin{equation*}
[s\oplus t]\pi:=\case{e}{([\inl s]\pi_1)}{([\inr t]\pi_1)} \mbox{ where $\pair{e,\pi_1}:=\omega\pi$}
\end{equation*}
In particular, defining $\flip:=\unit\oplus\unit:C+C$ we would have
\begin{equation*}
\sr{[\flip]}{\hr{\scount(n)}{H(c_i)\vee T(c_i)}{\scount(n+1)}}
\end{equation*}
although we stress that the operator $\oplus$ and would allow for much more complex probabilistic disjunctions, potentially involving additional computational content. 

Our soundness theorem, extended to these new probabilistic axioms and terms, would then facilitate the extraction of probabilistic programs from proofs. For instance, including a winner predicate $W(x)$, two player constant symbols $p_1,p_2$, and adding axioms 
\begin{equation*}
\begin{aligned}
H(c_1), H(c_2)&\mlpr\hr{\alpha}{W(p_1)}{\alpha}\\
T(c_1), T(c_2)&\mlpr\hr{\alpha}{W(p_1)}{\alpha}\\
H(c_1), T(c_2)&\mlpr\hr{\alpha}{W(p_2)}{\alpha}\\
T(c_1), H(c_2)&\mlpr\hr{\alpha}{W(p_2)}{\alpha}
\end{aligned}
\end{equation*}
for any $\alpha$, we could prove 
\begin{equation*}
\mlpr\hr{\scount(n)}{\exists x\, W(x)}{\scount(n+2)}
\end{equation*}
expressing the fact that a winner can be determined after two flips. We can then extract a corresponding probabilistic term for realizing this statement, which would be isomorphic to the expected program that queries the state twice in order to determine the outcome of those flips, and returns either $p_1$ or $p_2$ as a realizer for $\exists x\, W(x)$ depending on the content of the state.

Of course, the details here need to be worked through carefully in order to properly substantiate the claim that our framework could be used to extract probabilistic programs in a natural and meaningful way. At the very least, it is likely that further additions to $\ml$ along with a more intricate state would be needed to incorporate more interesting probabilistic events, such as annotated disjunctions along the lines of \cite{vennekens:etal:04:annotated}. We leave such matters to future work. 

\section*{Acknowledgments}
  \noindent The author thanks the anonymous referees who provided comments on earlier versions of this paper, which resulted in a much improved final version. This research was supported by the Engineering and Physical Sciences Research Council grant number EP/W035847/1.

  %% the following bibliography is gererated manually for the sake of brevity
  %% only; please use a separate .bib file in your submission

\bibliographystyle{alphaurl}
\bibliography{tpbiblio}

\end{document}